\newif\iflong
\longtrue  		

\iflong
\documentclass[12pt]{article}
\newcommand{\IEEEPARstart}[2]{#1#2}
\else
\documentclass[journal]{IEEEtran} 
\fi

\usepackage[dvipsnames, table]{xcolor}
\usepackage{url}
\usepackage{graphicx}
\usepackage{algorithmic, algorithm}
\usepackage[caption=false]{subfig}
\usepackage{amsmath}
\usepackage{amssymb}
\usepackage{multicol}
\usepackage{theorem}
\usepackage{booktabs}
\usepackage{multirow}
\usepackage{cite}

\newcommand{\scal}[2]{{\left\langle{{#1}\mid{#2}}\right\rangle}}
 
\newcommand{\menge}[2]{\big\{{#1}~\big |~{#2}\big\}} 

\newcommand{\HH}{\ensuremath{{\mathcal H}}}

\newcommand{\GG}{\ensuremath{{\mathcal G}}}
\newcommand{\emp}{\ensuremath{{\varnothing}}}
\newcommand{\epi}{\operatorname{epi}}

\newcommand{\Id}{\ensuremath{\operatorname{Id}}\,}
\newcommand{\RR}{\ensuremath{\mathbb{R}}}
\newcommand{\RP}{\ensuremath{\left[0,+\infty\right[}}

\newcommand{\RPP}{\ensuremath{\left]0,+\infty\right[}}

\newcommand{\RPX}{\ensuremath{\left[0,+\infty\right]}}
\newcommand{\RX}{\ensuremath{\left]-\infty,+\infty\right]}}

\newcommand{\NN}{\ensuremath{\mathbb N}}

\newcommand{\pinf}{\ensuremath{{+\infty}}}

\newcommand{\dom}{\ensuremath{\operatorname{dom}}}
\newcommand{\prox}{\ensuremath{\operatorname{prox}}}

 \newcommand{\reli}{\ensuremath{\operatorname{ri}}}
\newcommand{\minimize}[2]{\ensuremath{\underset{\substack{{#1}}}%
{\mathrm{minimize}}\;\;#2 }}

\newcommand{\argmind}[2]{\ensuremath{\underset{\substack{{#1}}}%
{\mathrm{argmin}}\;\;#2 }}

\newcommand{\pp}{p}
\newcommand{\qq}{q}
\newcommand{\rr}{r}
\newcommand{\PP}{{P'}}
\newcommand{\QQ}{{Q'}}
\renewcommand{\AA}{A'}
\newcommand{\BB}{B'}
\newcommand{\UU}{U}
\newcommand{\VV}{V}

\newcommand{\Prob}{\mathbb{P}}
  
\newtheorem{theorem}{Theorem}[section]
\newtheorem{lemma}[theorem]{Lemma}

\newtheorem{proposition}[theorem]{Proposition}

\theoremstyle{plain}{\theorembodyfont{\rmfamily}%
\newtheorem{example}[theorem]{Example}
\newtheorem{remark}[theorem]{Remark}}
\theoremstyle{plain}{\theorembodyfont{\rmfamily}%
\theoremstyle{plain}{\theorembodyfont{\rmfamily}%
}
\theoremstyle{plain}{\theorembodyfont{\rmfamily}%
}
\theoremstyle{plain}{\theorembodyfont{\rmfamily}%
\newtheorem{problem}[theorem]{Problem}}

\hyphenation{op-tical net-works semi-conduc-tor}

\newcommand{\authorOne}  {Mireille El Gheche}
\newcommand{\authorTwo}  {Giovanni Chierchia}
\newcommand{\authorThree}{Jean-Christophe Pesquet}

\newcommand{\affiliationOne}  {Universit\'e de Bordeaux, IMS UMR 5218 and IMB UMR 5251, Talence, France.}
\newcommand{\affiliationTwo}  {Universit\'e Paris-Est, LIGM UMR 8049, CNRS, ENPC, ESIEE Paris, UPEM, Noisy-le-Grand, France.}
\newcommand{\affiliationThree}{Center for Visual Computing, CentraleSup\'elec, University Paris-Saclay, Chatenay-Malabry, France.}

\begin{document}

\title{Proximity Operators of Discrete Information Divergences \iflong -- Extended Version\fi}

\iflong
\author{%
\authorOne\thanks{\affiliationOne}%
\and%
\authorTwo\thanks{\affiliationTwo}%
\and%
\authorThree\thanks{\affiliationThree}%
}
\fi


\maketitle

\begin{abstract}
\iflong
Information divergences allow one to assess how close two distributions are from each other. Among the large panel of available measures, a special attention has been paid to convex $\varphi$-divergences, such as Kullback-Leibler, Jeffreys, Hellinger, Chi-Square, Renyi, and I$_{\alpha}$ divergences.
\fi
While $\varphi$-divergences have been extensively studied in convex analysis, their use in optimization problems often remains challenging. In this regard, one of the main shortcomings of existing methods is that the minimization of $\varphi$-divergences is usually performed with respect to one of their arguments, possibly within alternating optimization techniques. In this paper, we overcome this limitation by deriving new closed-form expressions for the proximity operator of such two-variable functions. This makes it possible to employ standard proximal methods for efficiently solving a wide range of convex optimization problems involving $\varphi$-divergences. In addition, we show that these proximity operators are useful to compute the epigraphical projection of several functions. The proposed proximal tools are numerically validated in the context of optimal query execution within database management systems, where the problem of selectivity estimation plays a central role. Experiments are carried out on small to large scale scenarios. 
\end{abstract}

\section{Introduction}
\IEEEPARstart{D}{ivergence} measures play a crucial role in evaluating the dissimilarity between two information sources. The idea of quantifying how much information is shared between two probability distributions can be traced back to the work by Pearson \cite{Pearson1900} and Hellinger \cite{Hellinger1909}. Later, Shannon \cite{shannon48} introduced a powerful mathematical framework that links the notion of information with communications and related areas, laying the foundations for information theory. 
\iflong
However, information theory was not just a product of Shannon's work, it was the result of fundamental contributions made by many distinct individuals, from a variety of backgrounds, who took his ideas and expanded upon them. As a result, information theory has broadened to applications in statistical inference, natural language processing, cryptography, neurobiology, quantum computing, and other forms of data analysis. Important sub-fields of information theory are algorithmic information theory, information quantification, and source/channel coding. 
\fi
\textcolor{black}{In this context, a key measure of information is the Kullback-Leibler divergence~\cite{Kullback_1951_KL_divergences}, which can be regarded as an instance of the wider class of $\varphi$-divergences \cite{csiszar1963, ali_silvey_1966_jrs,bookelements_info_theory}, including also Jeffreys, Hellinger, Chi-square, R\'enyi, and $\mathrm{I}_\alpha$ divergences \cite{Sason2016}.}

\iflong

\subsection{Kullback-Leibler divergence}
The Kullback-Leibler (KL) divergence is known to play a prominent role in the computation of channel capacity and rate-distortion functions. One can address these problems with the celebrated alternating minimization algorithm proposed by Blahut and Arimoto \cite{Blahut72,Arimoto1972}. However, other approaches based on geometric programming may provide more efficient numerical solutions \cite{Chiang04geometricprogramming}. As the KL divergence is a Bregman distance, optimization problems involving this function can also be addressed by using the alternating minimization approach proposed by Bauschke \emph{et al} \cite{Bauschke06jointminimization} (see also \cite{Combettes_2016} for recent related works). However, the required optimization steps may be difficult to implement, and the convergence of the algorithm is only guaranteed under restrictive conditions. Moreover, a proximal algorithm generalizing the EM algorithm was investigated in \cite{citeulike:10288786}, where the KL divergence is a metric for maximizing a log-likelihood.

The generalized KL divergence (also called I-divergence) is widely used in inverse problems for recovering a signal of interest from an observation degraded by Poisson noise. In such a case, the generalized KL divergence is usually employed as a data fidelity term. The resulting optimization approach can be solved through an alternating projection technique \cite{Byrne1993}, where both the data fidelity term and the regularization term are based on the KL divergence. The problem was formulated in a similar manner by Richardson and Lucy \cite{RichardsonLucy72, Lucy74},  whereas more general forms of the regularization functions were considered by others \cite{Fessler1995, Dupe_FX_2008_ip_proximal_ifdpniusr, zanni_2009_invProb_Poisson_ProjMeth, anthoine_2010_Divergence_videoIndexing, Pustelnik_N_2011_tip_PPXA, TeuberSteidlChan2012, laurePoissonNoiszy_tip_2012, steidl_Pesquet_epigraphicalProj_2013}. In particular, some of these works are  grounded on proximal splitting methods \cite{Dupe_FX_2008_ip_proximal_ifdpniusr,Pustelnik_N_2011_tip_PPXA,TeuberSteidlChan2012}. These methods offer efficient and flexible solutions to a wide class of possibly nonsmooth convex minimization problems (see \cite{Combettes_P_2010_inbook_proximal_smsp,Boyd_proximalalgorithm_2013} and references therein). However, in all the aforementioned works, one of the two variables of the KL divergence is fixed.

\subsection{Other Divergences} 
Recently, the authors in \cite{nielsen_j_ieee_tmi_total_bregman_div, nielsen_p_cvpr_total_bregman_shapeRetrieval} defined a new measure called Total KL divergence, which has the benefit of being invariant to transformations from a special linear group. On the other side, the classical symmetrization of KL divergence, also known as Jeffreys-Kullback (JK) divergence \cite{jeffreys_1946_definition}, was recently used in the $k$-means algorithm as a replacement of the squared difference \cite{nielsen_2009_bregman_centroids,nielson_kullbackLeiber_centroid_2013}, yielding analytical expression of the divergence centroids in terms of the Lambert W function.

The Hellinger (Hel) divergence was originally introduced by Beran \cite{beran1977} and later rediscovered  under different names \cite{bookDevijver, gibbshellinger2002, liesevejda2006ieeeit, rauber2008prob}, such as Jeffreys-Masutita distance. In the field of information theory, the Hel divergence is commonly used for nonparametric density estimation \cite{lecam1973, sara1993}, statistics, and data analytics \cite{lee2012rule}, as well as machine learning~\cite{cieslak2011tree}. 

The Chi-square divergence was introduced by Pearson \cite{Pearson1900}, who used it to quantitatively assess whether an observed phenomenon tends to confirm or deny a given hypothesis. This work heavily contributed to the development of modern statistics. In 1984, the journal \emph{Science} referred to it as ``one of the 20 most important scientific breakthroughs''. Moreover, Chi-square was also successfully applied in different contexts, such as information theory and signal processing, as a dissimilarity measure between two probability distributions \cite{bookelements_info_theory,park2011chi}.   

R\'enyi divergence was introduced as a measure of information related to the R\'enyi entropy \cite{renyi1961mesureinfo}. According to the definition by Harremo \cite{harremo2006renyi_div}, R\'enyi divergence measures ``how much a probabilistic mixture of two codes can be compressed''. It has been studied and applied in many areas \cite{book_vajda_1989,lieseVajda1987_convstatdist,liesevejda2006ieeeit}, including image registration and alignement problems \cite{hero2002}.

The $\mathrm{I}_{\alpha}$ divergence was originally proposed by Chernoff \cite{chernoff1952} to statistically evaluate the efficiency of an hypothesis test. Subsequently, it was recognized as an instance of more general classes of divergences \cite{ali_silvey_1966_jrs}, such as the $\varphi$-divergences \cite{csiszar1974_infmeasure} and the Bregman divergences \cite{Amari2009}, and further extended by many researchers \cite{lieseVajda1987_convstatdist,zhangduality_div_2004,minka_div_2005_techReport,Amari2009}. The $\mathrm{I}_{\alpha}$ divergence has been also considered in the context of Non-negative Matrix Factorization, where the hyperparameter $\alpha$ is associated with characteristics of a learning machine \cite{Cichocki2008}.

\fi

\subsection{Contributions}
To the best of our knowledge, existing approaches for optimizing convex criteria involving $\varphi$-divergences are often restricted to specific cases, such as performing the minimization w.r.t.\ one of the divergence arguments. In order to take into account both arguments, one may resort to alternating minimization schemes, but only in the case when specific assumptions are met. Otherwise, there exist some approaches that exploit the presence of additional moment constraints \cite{csiszar_matus_momentconst_div}, or the equivalence between $\varphi$-divergences and some loss functions \cite{nguyen_lossfunctions_div}, but they provide little insight into the numerical procedure for solving the resulting optimization problems.

In the context of proximal methods, there exists no general approach for performing the minimization w.r.t.\ both the arguments of a $\varphi$-divergence. This limitation can be explained by the fact that a few number of closed-form expressions are available for the proximity operator of non-separable convex functions, as opposed to separable ones \cite{Chaux_C_2007_j-ip_variational_ffbip,Combettes_P_2010_inbook_proximal_smsp}. Some examples of such functions are the Euclidean norm~\cite{Combettes_PL_2008_j-ip_proximal_apdmfscvip}, the squared Euclidean norm composed with an arbitrary linear operator~\cite{Combettes_PL_2008_j-ip_proximal_apdmfscvip}, a separable function composed with an orthonormal or semi-orthogonal linear operator \cite{Combettes_PL_2008_j-ip_proximal_apdmfscvip}, the max function \cite{Condat_2014_simplex}, the quadratic-over-linear function \cite{raey,Benamou2000,Combettes2016perspective}, and the indicator function of some closed convex sets \cite{Combettes_PL_2008_j-ip_proximal_apdmfscvip, Chierchia_2012_Epigraphical_Projection}.

In this work, we develop a novel proximal approach that allows us to address more general forms of optimization problems involving $\varphi$-divergences. Our main contribution is the derivation of new closed-form expressions for the proximity operator of such functions. This makes it possible to employ standard proximal methods for efficiently solving a wide range of convex optimization problems involving $\varphi$-divergences. In addition to its flexibility, the proposed approach leads to parallel algorithms that can be efficiently implemented on both multicore and GPGPU architectures \cite{Gaetano2012}.

\subsection{Organization}
The remaining of the paper is organized as follows. Section~\ref{s:problemformulation} presents the general form of the optimization problem that we aim at solving.  Section~\ref{s:mainresult} studies the proximity operator of $\varphi$-divergences and some of its properties. Section~\ref{s:examples} details the closed-form expressions of the aforementioned proximity operators. Section~\ref{sec:epigraphicalProj} makes the connection with epigraphical projections. Section~\ref{s:results} illustrates the application to selectivity estimation for query optimization in database management systems. Finally, Section~\ref{sec:conclch3} concludes the paper.

\subsection{Notation}
Throughout the paper, $\Gamma_0(\HH)$ denotes the class of convex functions $f$ defined
on a real Hilbert space $\HH$ and taking their values in $]-\infty,\pinf~]$
which are lower-semicontinuous~and~proper (i.e. their domain $\dom f$ on which they take
finite values is nonempty).  $\|\cdot\|$ and $\scal{\cdot}{\cdot}$ denote the norm and the scalar product
of $\HH$, respectively. 
The Moreau subdifferential of $f$ at $x\in \HH$ is
$\partial f(x) = \menge{u\in\HH}{(\forall y\in\HH)\;\scal{y-x}{u}+f(x)\leq f(y)}$.
If  $f\in \Gamma_0(\HH)$ is G\^ateaux differentiable at $x$, $\partial f(x) = \{\nabla f(x)\}$
where $\nabla f(x)$ denotes the gradient of $f$ at $x$.
The conjugate of $f$ is $f^*\in \Gamma_0(\HH)$ such that
$(\forall  u\in \HH)$ $f^*(u) = \sup_{x\in\HH}\big(\scal{x}{u}-f(x)\big)$.
\textcolor{black}{The proximity operator of $f$  is the mapping $\prox_f\colon\HH\to\HH$ defined as \cite{Moreau_J_1965_bsmf_Proximite_eddueh}
	\begin{equation}
	(\forall x \in \HH)\qquad \prox_f(x) = \argmind{y\in\HH}{f(y)+\frac12 \|x-y\|^2}.
	\end{equation}}%
Let $C$ be a nonempty closed convex subset $C$ of $\HH$.
The indicator function of  $C$ is defined as
\begin{equation}
(\forall x \in \HH)\qquad \iota_{C}(x) = 
\begin{cases}
0 & \mbox{if $x \in C$}\\
\pinf & \mbox{otherwise.}
\end{cases}
\end{equation}
The elements of a vector $x\in\HH = \RR^N$ are denoted by $x=(x^{(\ell)})_{1\le \ell\le N}$, whereas $I_N$ is the $N\times N$ identity matrix.

\section{Problem formulation}
\label{s:problemformulation}

The objective of this paper is to address convex optimization problems involving
a discrete information divergence. In particular, the focus is put on the following formulation.

\begin{problem}\label{p:gen_div}
	Let $D$ be a function in $\Gamma_0(\RR^P\times \RR^P)$.
	Let $A$ and $B$ be matrices in $\RR^{P\times N}$, and let $u$ and $v$
	be vectors in $\RR^P$. For every $s\in \{1,\ldots,S\}$,
	let $R_s$ be a function in $\Gamma_0(\RR^{K_s})$ and $T_s \in \RR^{K_s
		\times N}$.
	We want to
	\begin{equation}
	\label{prob:probgen}
	\minimize{x\in \RR^N}{D(A x+u,B x+v) + \sum_{s=1}^S R_s(T_s x)}.
	\end{equation}
\end{problem}

Note that the functions $D$ and $(R_s)_{1\le s \le S}$ are allowed to take the value $\pinf$, so that Problem \ref{p:gen_div} can include convex constraints by letting some of the functions $R_s$ be equal to the indicator function $\iota_{C_s}$ of some nonempty closed convex set $C_s$. In inverse problems, $R_s$ may also model some additional prior information, such as the sparsity of coefficients after some appropriate linear transform $T_s$.

{\color{black}
\subsection{Applications in information theory}
A special case of interest in information theory arises by decomposing $x$ into two vectors $\pp\in\RR^\PP$ and $\qq\in\RR^\QQ$, that is $x = [\pp^\top\;\qq^\top]^\top$ with $N=\PP+\QQ$. 
Indeed, set $u=v = 0$, $A = [\AA\; 0]$ with $\AA \in \RR^{P \times \PP}$, $B = [0\; \BB]$ with $\BB\in \RR^{P\times \QQ}$ and, for every $s \in \{1,\ldots,S\}$,
$T_s = [\UU_s\; \VV_s]$ with $\UU_s \in \RR^{K_s \times \PP}$ and $\VV_s \in \RR^{K_s \times \QQ}$. Then, Problem~\ref{p:gen_div} takes the following form:
\vspace*{-0.2cm}
\begin{problem}\label{p:prob}
Let $\AA$, $\BB$, $(\UU_s)_{1\le s \le S}$, and $(\VV_s)_{1\le s \le S}$ be matrices as defined above.
Let $D$ be a function in $\Gamma_0(\RR^P\times \RR^P)$ and, for every $s\in \{1,\ldots,S\}$, let $R_s$ be a function in $\Gamma_0(\RR^{K_s})$.
We want to
\begin{equation}
\minimize{(\pp,\qq)\in \RR^\PP\times\RR^\QQ}{D(\AA\pp,\BB\qq) + \sum_{s=1}^S
R_s(\UU_s\pp+\VV_s\qq)}.
\end{equation}
\end{problem}
Several tasks can be formulated within this framework, such as the computation of  channel  capacity  and  rate-distortion  functions  \cite{Arimoto1972, Blahut72}, the selection of log-optimal portfolios \cite{Cover1984}, maximum likelihood estimation from incomplete data \cite{Dempster1977}, soft-supervised learning for text classification \cite{Subramanya2008}, simultaneously estimating a regression vector and an additional model parameter \cite{Combettes2016perspective} or the image gradient distribution and a parametric model distribution \cite{Tartavel2016}, as well as image registration \cite{ElGheche_2013_EUSIPCO}, deconvolution \cite{ElGheche_2013_ICASSP}, and recovery~\cite{Byrne1993}.
We next detail an important application example in source coding.
	
	\begin{example}
		Assume that a discrete memoryless source $E$, taking its values in a finite alphabet $\{e_1,\dots,e_{P_1}\}$ with probability $\Prob(E)$, is to be encoded by a compressed signal $\widehat{E}$ in terms of a second alphabet $\{\widehat{e}_1,\dots,\widehat{e}_{P_2}\}$. Furthermore, for every $j\in \{1,\ldots,P_1\}$ and $k \in \{1,\ldots,P_2\}$, let 
		$\delta^{(k,j)}$ be the distortion induced when substituting $\widehat{e}_k$ for $e_j$.
		We wish to find an encoding $\Prob(\widehat{E}|E)$ that yields a point on the 
		rate-distortion curve at a given distortion value $\overline{\delta} \in\RPP$.
		It is well-known \cite{Cover2006} that this amounts to minimizing the mutual information $\mathcal{I}$ between $E$ and $\widehat{E}$, more precisely
		the rate-distortion function $\mathsf{R}$ is given by
		\begin{equation}\label{eq:RD_fun}
		\mathsf{R}(\overline{\delta}) = \min_{\Prob(\widehat{E}|E)} \mathcal{I}(E,\widehat{E}), %
		\end{equation}
		subject to the constraint
		\begin{equation}\label{eq:RD_set}
		\sum_{j=1}^{P_1}\sum_{k=1}^{P_2} \delta^{(k,j)} \, \Prob(E=e_j)\Prob(\widehat{E}=\widehat{e}_k|E=e_j) \le \overline{\delta}.
		\end{equation}
		The mutual information can be written as \cite[Theorem~4(a)]{Blahut72}
		\begin{equation}
		\min_{\Prob(\widehat{E})} \; \sum_{j=1}^{P_1}\sum_{k=1}^{P_2}\Prob(E=e_j,\widehat{E}=\widehat{e}_k) \ln\left( \frac{\Prob(\widehat{E}=\widehat{e}_k,E=e_j)}{\Prob(E=e_j)\Prob(\widehat{E}=\widehat{e}_k)}\right),
		\end{equation}
		subject to the constraint 
		\begin{equation}\label{e:C1RD}
		\sum_{k=1}^{P_2} \Prob(\widehat{E}=\widehat{e}_k) = 1.
		\end{equation}
		Moreover, the constraint in \eqref{eq:RD_set} can be reexpressed as
		\begin{equation}\label{e:C2RD}
		\sum_{j=1}^{P_1}\sum_{k=1}^{P_2} \delta^{(k,j)} \, \Prob(E=e_j,\widehat{E}=\widehat{e}_k)   \le \overline{\delta},
		\end{equation}
		with
		\begin{equation}\label{e:C3RD}
		(\forall j\in \{1,\ldots,P_1\})\quad \sum_{k=1}^{P_2}  \Prob(E=e_j,\widehat{E}=\widehat{e}_k)  = \Prob(E=e_j).
		\end{equation}
		The unknown variables are thus the vectors 
		\begin{equation}
		\pp=\big(\Prob(E=e_j,\widehat{E}=\widehat{e}_k)\big)_{1\le j\le P_1,1\le k\le P_2} \in \RR^{P_1 P_2}
		\end{equation}
		and
		\begin{equation}
		\qq =\big(\Prob(\widehat{E}=\widehat{e}_k)\big)_{1\le k\le P_2} \in \RR^{P_2},
		\end{equation}
		whose optimal values  are solutions to the problem:
		\begin{equation}\label{eq:RD}
		\minimize{\pp \in C_2\cap C_3, \qq \in C_1} D(\pp,\rr\otimes\qq)
		\end{equation}
		where $\rr = \big(\Prob(E=e_j)\big)_{1\le j\le P_1} \in \RR^{P_1}$,
		$\otimes$ denotes the Kronecker product, $D$ is the Kullback-Leibler divergence, and
		$C_1$, $C_2$, $C_3$ are the closed convex sets corresponding to the linear constraints \eqref{e:C1RD}, \eqref{e:C2RD}, \eqref{e:C3RD},
		respectively. The above formulation is a special case of Problem \ref{p:prob} in which $P=P'=P_1 P_2$, $Q' = P_2$, $\AA=I_P$, $\BB$ is such that  
		$(\forall \qq \in \RR^{Q'})$ $\BB\qq=\rr\otimes\qq$, $S=3$, $\VV_1 = I_{Q'}$, $\UU_2 = \UU_3 = I_P$,
		$\UU_1$ and $\VV_2=\VV_3$ are null matrices, and
		$(\forall s \in \{1,2,3\})$ $R_s$ is the indicator function of the constraint convex set $C_s$.

	\end{example}
	
}

\subsection{Considered class of divergences}\label{se:defdiv}
We will focus on additive information measures of the form
\begin{equation}\label{e:defD}
\big(\forall (p,q) \in \RR^P\times\RR^P\big)\qquad D(p,q) = \sum_{i=1}^P \Phi(p^{(i)},q^{(i)}),
\end{equation} 
where $\Phi\in\Gamma_0(\RR\times\RR)$ is the \emph{perspective function} \cite{Bauschke_H_2011_book_con_amo} on $\RP\times \RPP$ of a function $\varphi\colon\RR \to \RPX$ belonging to $\Gamma_0(\RR)$ and twice differentiable on $\RPP$. In other words, $\Phi$ is defined as follows: for every $(\upsilon,\xi)\in \RR^2$,
\begin{equation}
\Phi(\upsilon,\xi) = 
\left\{
\begin{aligned}
&\xi\, \varphi\Big(\frac{\upsilon}{\xi}\Big) && \mbox{if $\upsilon \in \RP$ and $\xi \in \RPP$}\\
&\upsilon \lim_{\zeta \to \pinf} \frac{\varphi(\zeta)}{\zeta} &&
\mbox{if $\upsilon \in \RPP$ and $\xi = 0$}\\
&0 && \mbox{if $\upsilon = \xi = 0$}\\
&\pinf && \mbox{otherwise},
\end{aligned}
\right.
\label{e:perspf}
\end{equation}
where the above limit is guaranteed to exist \cite[Sec.~2.3]{Hiriart_Urruty_1996_book_convex_amaIf}.
Moreover, if $\varphi$ is a strictly convex function such that 
\begin{equation}
\varphi(1) = \varphi'(1) = 0,
\end{equation}
the function $D$ in \eqref{e:defD} belongs to the class of
$\varphi$-divergences \cite{csiszar1963,Csiszar_1967_divergences}.
Then, for every $(p,q) \in \RP^P \times \RP^P$,
\begin{align}
& D(p,q) \ge 0 &&\\
& D(p,q) = 0 \quad \Leftrightarrow \quad p=q.&&
\end{align}
Examples of $\varphi$-divergences will be provided in Sections
\ref{se:DKL}, \ref{se:DJK}, \ref{se:DH}, \ref{se:DCS} and \ref{se:DIa}.
For a thorough investigation of the rich properties of $\varphi$-divergences, the reader is refered to \cite{csiszar1963,ali_silvey_1966_jrs, basseville_1089_distanceMeasures_signal_pattern_recog}.
Other divergences (e.g., R\'enyi divergence) are expressed as
\begin{equation}
\big(\forall (p,q) \in \RR^P \times \RR^P\big)\qquad
D_g(p,q) = g\big(D(p,q)\big)
\end{equation} 
where $g$ is an increasing function. Then, provided that $g\big(\varphi(1)\big) = 0$, $D_g(p,q) \ge 0$ for every $[p^\top \; q^\top]^\top\in C$ with
\begin{equation}\label{e:Cprob}
C = \Big\{{x \in [0,1]^{2P}} ~\big|~ {\sum_{i=1}^P x^{(i)} = 1\;\text{and}\;\sum_{i=1}^{P} x^{(P+i)} = 1}\Big\}.
\end{equation}
From an optimization standpoint,
minimizing $D$ or $D_g$ (possibly subject to constraints)  makes no
difference, hence we will only address problems involving $D$ in the
rest of this paper.

\subsection{Proximity operators}
Proximity operators will be fundamental tools in this paper.
We first recall some of their key properties.

\begin{proposition} 
	\label{p:propprox}
	{\rm \cite{Moreau_J_1965_bsmf_Proximite_eddueh,Bauschke_H_2011_book_con_amo}} Let $f\in \Gamma_0(\HH)$. Then,
	\begin{enumerate}
		\item \label{p:propproxi} For every $\overline{x}\in \HH$, $\prox_f \overline{x} \in \dom
		f$.
		\item \label{p:propproxii} For every $(x,\overline{x})\in \HH^{2}$
		\begin{equation}
		x = \prox_f(\overline{x})\quad\Leftrightarrow\quad \overline{x} - x \in \partial f(x).
		\end{equation}
		\item \label{p:propproxiii} For every $(\overline{x},z)\in \HH^2$,
		\begin{equation}
		\prox_{f(\cdot+z)}(\overline{x}) = \prox_f(\overline{x}+z)-z.
		\end{equation}
		\item \label{p:propproxiv} For every $(\overline{x},z)\in \HH^2$ and for every $\alpha \in
		\RR$,
		\begin{equation}
		\prox_{f+ \scal{z}{\cdot} + \alpha} (\overline x) =
		\prox_{f}(\overline{x}-z).
		\end{equation}
		\item \label{p:propproxvi} Let $f^*$ be the conjugate function of $f$. For every $\overline{x} \in \HH$ and for every $\gamma \in \RPP$, 
		\begin{equation}
		\prox_{\gamma f^*} (\overline{x}) = \overline{x} - \gamma\prox_{f/\gamma}(\overline{x}/\gamma).
		\end{equation}
		\item \label{p:propproxv} Let $\GG$ be a real Hilbert space and let $T\colon \GG\to \HH$ be a bounded linear operator, with the adjoint denoted by $T^*$. If $T T^* = \kappa \Id$ and $\kappa \in \RPP$, then for all $\overline{x} \in \HH$
		\begin{equation}
		\prox_{f \circ T}(\overline{x}) = \overline{x} + \frac{1}{\kappa} T^*\big(\prox_{\kappa f}(T\overline{x}) - T\overline{x}\big).
		\end{equation}
	\end{enumerate}
\end{proposition}
Numerous additional properties of proximity operators are mentioned in \cite{Combettes_P_2010_inbook_proximal_smsp, Boyd_proximalalgorithm_2013}.

In this paper, we will be mainly concerned with the determination of
the proximity operator of the function $D$ defined in \eqref{e:defD} with $\HH = \RR^P \times \RR^P$. The next result emphasizes that this task reduces to the calculation of the proximity operator of a
real function of two variables.
\begin{proposition}\label{p:proxD}
	Let $D$ be defined by \eqref{e:defD} where $\Phi \in \Gamma_0(\RR^2)$ and let $\gamma \in \RPP$. Let $u \in \RR^P$ and $v \in \RR^P$. Then, for every $\overline{p}\in \RR^P$ and for every $\overline{q} \in \RR^P$,
	\begin{equation}
	\prox_{\gamma D(\cdot+u,\cdot+v)}(\overline{p},\overline{q}) = (p-u, q-v)
	\end{equation}
	where, for every $i \in \{1,\ldots,P\}$, 
	\begin{align}
	(p^{(i)},q^{(i)}) = \prox_{\gamma \Phi}(\overline{p}^{(i)}+u^{(i)},\overline{q}^{(i)}+v^{(i)}).
	\end{align}
\end{proposition}
\iflong
\begin{proof}
	The result is a straightforward consequence of \cite[Table 10.1ix]{Combettes_P_2010_inbook_proximal_smsp} and
	Proposition~\ref{p:propprox}\ref{p:propproxiii}, by setting $f = D$
	and $z = (u,v)$.
\end{proof}
\fi

Note that, although an extensive list of proximity operators of
one-variable real functions can be found in \cite{Combettes_P_2010_inbook_proximal_smsp}, few results are
available for real functions of two variables \cite{raey,Benamou2000,Combettes_PL_2008_j-ip_proximal_apdmfscvip,Chierchia_2012_Epigraphical_Projection}. An example of such a result
is provided below.
\begin{proposition}\label{p:proxdifsimp}
	Let $\varphi \in \Gamma_0(\RR)$ be an even differentiable function on $\RR\setminus \{0\}$.
	Let $\Phi\colon \RR^2 \to \RX$ be defined as: $(\forall (\nu,\xi)\in \RR^2)$
	\begin{equation}
	\Phi(\nu,\xi) = 
	\left\{
	\begin{aligned}
	&\varphi(\nu-\xi) &&\textrm{if $(\nu,\xi) \in \RP^2$}\\
	&+\infty &&\textrm{otherwise}.
	\end{aligned}
	\right.
	\end{equation}
	Then, for every $(\overline{\nu},\overline{\xi})\in \RR^2$,
	\begin{gather}
	\!\!\!\!\!\!
	\prox_{\Phi}(\overline{\nu},\overline{\xi}) = 
	\left\{
	\begin{aligned}
	&\frac12\big(\overline{\nu}+\overline{\xi}+\pi_1,\overline{\nu}+\overline{\xi}-\pi_1\big) && \mbox{if $|\pi_1| < \overline{\nu}+\overline{\xi}$}\\
	&(0,\pi_2) && \mbox{if $\pi_2 > 0$ and $\pi_2 \ge \overline{\nu}+\overline{\xi}$}\\
	&(\pi_3,0) && \mbox{if $\pi_3 > 0$ and $\pi_3 \ge \overline{\nu}+\overline{\xi}$}\\
	&(0,0) && \mbox{otherwise,}
	\end{aligned}
	\right.
	\end{gather}
	with $\pi_1 = \prox_{2\varphi}(\overline{\nu}-\overline{\xi})$, $\pi_2=\prox_{\varphi}(\overline{\xi})$ and $\pi_3=\prox_{\varphi}(\overline{\nu})$.
\end{proposition}
\iflong
\begin{proof} See Appendix \ref{a:proxdifsimp}
\end{proof}
\fi

The above proposition provides a simple characterization of the proximity operators of some distances defined for nonnegative-valued vectors.
However, the assumptions made in Proposition \ref{p:proxdifsimp} are not satisfied by the class of functions $\Phi$ considered in Section \ref{se:defdiv}.\footnote{Indeed, none of the considered $\varphi$-divergences can be expressed as a function of the difference between the two arguments.} In the next section, we will propose two algorithms for solving a general class of convex problems involving these functions $\Phi$.

\subsection{Proximal splitting algorithms}\label{sec:proximal_methods}
As soon as we know how to calculate the proximity operators of the
functions involved in Problem~\ref{p:gen_div}, various proximal methods
can be employed to solve it numerically. Two examples of such methods
are given subsequently.

The first algorithm is PPXA+ \cite{Pesquet_J_2012_j-pjpjoo_par_ipo} which constitutes an extension of
PPXA (Parallel ProXimal Agorithm) proposed in \cite{Combettes_PL_2008_j-ip_proximal_apdmfscvip}.  
As it can be seen in \cite{Afonso_M_2009_j-tip_augmented_lacofiip,Setzer_S_2009_j-jvcir_deblurring_pibsbt}, PPXA+ is an augmented Lagrangian-like methods (see also \cite[Sec.\ 6]{Pesquet_J_2012_j-pjpjoo_par_ipo}).

\begin{minipage}{0.9\linewidth}
	\begin{algorithm}[H]
		\caption{PPXA+}
		\label{e:PPXA+}
		\small
		\vspace{0.5em}
		\textsc{Initialization}\\[-0.5em]
		\[
		\left\lfloor
		\begin{aligned}
		& \textrm{$(\omega_0,\ldots,\omega_S) \in \RPP^{S+1},$}\\
		& \textrm{$(t_{0,0}, t_{1,0})\in \RR^P\times\RR^P, t_{2,0}\in\RR^{K_1},\ldots,t_{S+1,0} \in \RR^{K_S}$}\\
		& \textrm{$\textsf{Q} = \Big(\omega_0 A^\top A + \omega_0 B^\top B + \sum_{s=1}^S \omega_s T_s^\top T_s\Big)^{-1}$}\\
		& \textrm{$x_0 = \textsf{Q}\Big(\omega_0 A^\top t_{0,0}+ \omega_0 B^\top t_{1,0} + \sum_{s=1}^S \omega_s T_s^\top t_{s+1,0}\Big)$.}\\
		\end{aligned}
		\right.
		\]
		\noindent \textsc{For}\; $n = 0, 1, \dots$\\[-0.5em]
		\[
		\left\lfloor
		\begin{aligned}
		& \textrm{$(r_{0,n},r_{1,n}) = \prox_{\omega_0^{-1} D(\cdot+u,\cdot+v)}(t_{0,n},t_{1,n})+e_{0,n}$}\\	
		& \textrm{For \; $s = 0, 1, \dots S$}\\
		& \; \left\lfloor\begin{aligned}
		& \textrm{$r_{s+1,n}=\prox_{\omega_s^{-1} R_s} (t_{s+1,n})+e_{s,n}$}
		\end{aligned}\right. \\
		&\textrm{ $y_n =\textsf{Q}\Big(\omega_0 A^\top r_{0,n}+ \omega_0 B^\top r_{1,n} + \sum_{s=1}^S \omega_s T_s^\top r_{s+1,n}\Big)$}\\
		&\textrm{ $\lambda_n\in\left]0,2\right[$}\\
		&\textrm{ $t_{0,n+1}=t_{0,n}+\lambda_n\Big(A(2y_n -x_n)-r_{0,n}\Big)$}\\
		&\textrm{ $t_{1,n+1}=t_{1,n}+\lambda_n\Big(B(2y_n -x_n)-r_{1,n}\Big)$}\\	
		& \textrm{For \; $s = 0, 1, \dots S$}\\
		& \; \left\lfloor\begin{aligned}
		& \textrm{$t_{s+1,n+1}=t_{s+1,n}+\lambda_n\Big(T_s(2y_n -x_n)-r_{s+1,n}\Big)$}
		\end{aligned}\right. \\
		&\textrm{$x_{n+1}=x_n+\lambda_n(y_{n}-x_n).$ }				
		\end{aligned}
		\right.
		\]
	\end{algorithm}
	\vspace{0.1cm}
\end{minipage}

In this algorithm, $\omega_0,\ldots,\omega_S$ are weighting factors and $(\lambda_n)_{n\ge 0}$ are relaxation factors. For every $n\ge 0$, the variables $e_{0,n}\in\RR^P\times\RR^P$, $e_{1,n}\in\RR^{K_1}$, \dots, $e_{S,n}\in \RR^{K_S}$ model possible errors in the computation of the proximity operators. For instance, these errors arise when the proximity operator is not available in a closed form, and one needs to compute it through inner iterations. Under some technical conditions, the convergence of PPXA+ is guaranteed.

\begin{proposition}{\rm \cite[Corollary 5.3]{Pesquet_J_2012_j-pjpjoo_par_ipo}}
	Suppose that the following assumptions hold.
	\begin{enumerate}
		\item The matrix $\displaystyle A^\top A + B^\top B + \sum_{s=1}^S T_s^\top T_s$ is invertible.
		\item
		\label{p:3ii} There exists $\check{x} \in \RR^N$ such that
		\begin{equation}\label{eq:cond}
		\begin{cases}
		A \check{x}+u \in\RPP^P\\
		B\check{x}+v \in \RPP^P\\
		(\forall s \in \{1,\ldots,S\})\quad T_s\check{x} \in \reli(\dom R_s).
		\end{cases}
		\end{equation}
		\item
		\label{p:3iii}
		There exists
		$\underline{\lambda} \in ]0,2[$ such that, for every $n \in \NN$,
		\begin{equation}
		\underline{\lambda} \le \lambda_{n+1} \le \lambda_n < 2.
		\end{equation}
		\item
		\label{p:3iv}
		For every $s\in\{0,\ldots,S\}$, 
		\begin{equation}
		\sum_{n\in\NN}\|e_{s,n}\|<\pinf.
		\end{equation}
	\end{enumerate}
	If the set of solutions to Problem \ref{p:gen_div}
	is nonempty, then any sequence $(x_n)_{n\in \NN}$ generated by
	Algorithm~\ref{e:PPXA+} converges to an element of this set.
\end{proposition}

It can be noticed that, at each iteration $n$, PPXA+ requires to solve
a linear system in order to compute the intermediate variable $y_n$.
The computational cost of this operation may be high when $N$ is
large. Proximal primal-dual approaches \cite{Chen_G_1994_j-mp_pro_bdm, Esser_E_2010_j-siam-is_gen_fcf, Chambolle_A_2010_first_opdacpai, Briceno_L_2011_j-siam-opt_mon_ssm, Combettes_P_2011_j-svva_pri_dsa, Vu_B_2011_j-acm_spl_adm, Condat_L_2012, Komodakis2015} allow us to circumvent
this difficulty. An example of such an approach is the 
Monotone+Lipschitz Forward Backward Forward (M+LFBF) method \cite{Combettes_P_2011_j-svva_pri_dsa} which
takes the following form.

\begin{minipage}{0.9\linewidth}
	\begin{algorithm}[H]
		\caption{M+LFBF}
		\label{e:FBF}
		\small
		\vspace{0.5em}
		\textsc{Initialization}\\[-0.5em]
		\[
		\left\lfloor
		\begin{aligned}
		& \textrm{$(t_{0,0}, t_{1,0})\in \RR^P\times\RR^P, t_{2,0}\in\RR^{K_1},\ldots,t_{S+1,0} \in \RR^{K_S}$}\\
		& \textrm{$x_0\in \RR^N, \; \beta = \Big(\|A\|^2+\|B\|^2+\sum_{s=1}^S \|T_s\|^2\Big)^{1/2},\;$}\\
		& \textrm{$\varepsilon\in\left]0,1/(\beta+1)\right[$}.\\
		\end{aligned}
		\right.
		\]
		\noindent \textsc{For}\; $n = 0, 1, \dots$\\[-0.5em]
		\[
		\left\lfloor
		\begin{aligned}
		& \textrm{$\gamma_n \in [\varepsilon,(1-\varepsilon)/\beta]$}\\
		& \textrm{$\widehat{x}_n=x_n-\gamma_n(A^\top t_{0,n}+ B^\top t_{1,n} + \sum_{s=1}^S T_s^\top t_{s+1,n})$}\\
		& \textrm{$\big(\widehat{t}_{0,n},\widehat{t}_{1,n}\big)  = \big( t_{0,n}, t_{1,n} \big) +\gamma_n \big(A x_n, B x_n\big)$} \\
		& \textrm{$(r_{0,n},r_{1,n}) = (\widehat{t}_{0,n},\widehat{t}_{1,n})-$}\\
		& \textrm{$\qquad \qquad\gamma_n \prox_{\gamma_n^{-1} D(\cdot+u,\cdot+v)}(\gamma_n^{-1}\widehat{t}_{0,n},\gamma_n^{-1}\widehat{t}_{1,n})+e_{0,n}$}\\
		& \textrm{$\big(\widetilde{t}_{0,n}, \widetilde{t}_{1,n}\big)=\big(r_{0,n},r_{1,n}\big) +\gamma_n \big(A \widehat{x}_n, B \widehat{x}_n\big)$}\\ 
		& \textrm{$\big(t_{0,n+1},t_{1,n+1}\big) = \big(t_{0,n},t_{1,n}\big)-\big(\widehat{t}_{0,n},\widehat{t}_{1,n}\big)+\big(\widetilde{t}_{0,n},\widetilde{t}_{1,n}\big)$}\\
		& \textrm{For \; $s = 0, 1, \dots S$}\\
		& \; \left\lfloor\begin{aligned}
		& \textrm{$\widehat{t}_{s+1,n} = t_{s+1,n}+\gamma_n T_s x_n$}\\
		& \textrm{$r_{s+1,n}= \widehat{t}_{s+1,n}-\gamma_n \prox_{\gamma_n^{-1} R_s}(\gamma_n^{-1}\widehat{t}_{s+1,n})+e_{s,n}$}\\
		& \textrm{$\widetilde{t}_{s+1,n}=r_{s+1,n}+\gamma_n T_s \widehat{x}_n$}\\
		& \textrm{$t_{s+1,n+1}=t_{s+1,n}-\widehat{t}_{s+1,n}+\widetilde{t}_{s+1,n}$}\\
		\end{aligned}\right. \\
		& \textrm{$\widetilde{x}_{n}=\widehat{x}_{n}-\gamma_n (A^\top r_{0,n}+ B^\top r_{1,n} + \sum_{s=1}^S T_s^\top r_{s+1,n})$}\\
		& \textrm{$x_{n+1}=x_n-\widehat{x}_n+\widetilde{x}_n.$}\\
		\end{aligned}
		\right.
		\]
	\end{algorithm}
	\vspace{0.1cm}
\end{minipage}

In this algorithm, $(\gamma_n)_{n\ge 0}$ is a sequence of step-sizes, and $e_{0,n}\in\RR^P\times\RR^P$, $e_{1,n}\in\RR^{K_1}$, \dots, $e_{S,n}\in \RR^{K_S}$ correspond to possible errors in the computation of proximity operators. The convergence is secured by the following result.

\begin{proposition}{\rm \cite[Theorem 4.2]{Combettes_P_2011_j-svva_pri_dsa}}
	Suppose that the following assumptions hold.
	\begin{enumerate}
		\item There exists $\check{x} \in \RR^N$ such that \eqref{eq:cond} holds.
		\item
		$(\forall s\in\{0,\ldots,S\})$
		$\sum_{n\in\NN}\|e_{s,n}\|<\pinf$.
	\end{enumerate}
	If the set of solutions to Problem \ref{p:gen_div}
	is nonempty, then any sequence $(x_n)_{n\in \NN}$ generated by 
	Algorithm~\ref{e:FBF} converges to an element of this set.
\end{proposition}

It is worth highlighting that these two algorithms share two interesting features: many operations can be implemented in parallel (e.g., the loops on $s$), there is a tolerance to errors in the computation of the proximity operators. Recently, random block-coordinate versions of proximal algorithms have been proposed (see \cite{Pesquet2015_Random} and references therein) further improving the flexibility of these methods.

\section{Main result}
\label{s:mainresult}
As shown by Proposition \ref{p:proxD},
we need to compute the proximity operator of a scaled version of
a function $\Phi \in \Gamma_0(\RR^2)$ as defined in \eqref{e:perspf}.
In the following, $\Theta$ denotes a primitive on $\RPP$
of the function $\zeta \mapsto \zeta \varphi'(\zeta^{-1})$.
The following functions will subsequently play an important role: 
\begin{align}
&\vartheta_{-}\colon \RPP \to \RR\colon \zeta
\mapsto\varphi'(\zeta^{-1})\label{e:defvarthetam}\\
&\vartheta_{+}\colon \RPP \to \RR\colon \zeta \mapsto 
\varphi(\zeta^{-1}) - \zeta^{-1}\varphi'(\zeta^{-1}).
\label{e:defvarthetap}
\end{align}
A first technical result is as follows.
\begin{lemma}\label{le:psi}
	Let $\gamma \in \RPP$, let $(\overline{\upsilon},\overline{\xi})\in
	\RR^2$,  and define
	\begin{align}
	&\chi_{-} = \inf\menge{\zeta \in \RPP}{\vartheta_{-}(\zeta) < \gamma^{-1}
		\overline{\upsilon}} \label{e:defchim}\\
	&\chi_{+} = \sup
	\menge{\zeta\in \RPP}{\vartheta_{+}(\zeta) < \gamma^{-1} \overline{\xi}}\label{e:defchip}
	\end{align}
	(with the usual convention $\inf \emp = \pinf$ and $\sup \emp = -\infty$). If $\chi_{-} \neq \pinf$, the function
	\begin{align}
	\psi\colon& \RPP \to \RR\colon\nonumber\\
	&\zeta \mapsto
	\zeta \varphi(\zeta^{-1})-\Theta(\zeta) +
	\frac{\gamma^{-1}\overline{\upsilon}}{2}\zeta^2 - \gamma^{-1}
	\overline{\xi} \zeta
	\end{align}
	is strictly convex on $]\chi_{-},\pinf[$.
	In addition, if 
	\begin{enumerate}
		\item\label{as:lepsi0} $\Big.\chi_{-} \neq \pinf$ and $\chi_{+} \neq - \infty$
		\item\label{as:lepsi1} $\Big.\lim_{\substack{\zeta \to \chi_{-}\\\zeta > \chi_{-}}} \psi'(\zeta) < 0$
		\item\label{as:lepsi2} $\Big.\lim_{\zeta \to \chi_{+}} \psi'(\zeta) > 0$
	\end{enumerate}
	then $\psi$ admits a unique minimizer $\widehat{\zeta}$ on $]\chi_{-},\pinf[$, and 
	$\widehat{\zeta} \!<\! \chi_{+}$. 
\end{lemma}
\begin{proof}
	The derivative of $\psi$ is, for every $\zeta \in \RPP$,
	\begin{align}\label{e:derpsi}
	\psi'(\zeta) &= \varphi(\zeta^{-1})-(\zeta+\zeta^{-1})
	\varphi'(\zeta^{-1}) + \gamma^{-1} \overline{\upsilon} \zeta -
	\gamma^{-1} \overline{\xi}\nonumber\\
	& = \zeta\big(\gamma^{-1} \overline{\upsilon} - \vartheta_{-}(\zeta)\big)
	+ \vartheta_{+}(\zeta)-\gamma^{-1} \overline{\xi}.
	\end{align}
	The function $\vartheta_{-}$ is decreasing as the convexity of $\varphi$ yields
	\begin{equation}
	(\forall \zeta \in \RPP)\qquad
	\vartheta'_{-}(\zeta) = -\zeta^{-2} \varphi''(\zeta^{-1}) \le 0.
	\end{equation}
	This allows us to deduce that
	\begin{align}\label{e:interchim}
	&\text{if}\;\menge{\zeta \in \RPP}{\vartheta_{-}(\zeta) < \gamma^{-1} \overline{\upsilon}}
	\neq \emp,\nonumber\\
	&\text{then}\;]\chi_{-},\pinf[ =
	\menge{\zeta \in \RPP}{\vartheta_{-}(\zeta) < \gamma^{-1} \overline{\upsilon}}.
	\end{align}
	Similarly, the function $\vartheta_{+}$ is
	increasing as the convexity of $\varphi$ yields
	\begin{equation}
	(\forall \zeta \in \RPP)\qquad
	\vartheta'_{+}(\zeta) = \zeta^{-3} \varphi''(\zeta^{-1}) \ge 0
	\end{equation}
	which allows us to deduce that
	\begin{align} \label{e:interchip}
	&\text{if}\;\menge{\zeta \in
		\RPP}{\vartheta_{+}(\zeta) < \gamma^{-1} \overline{\xi}}\neq \emp,\nonumber\\
	&\text{then}\;]0,\chi_{+}[= \menge{\zeta \in
		\RPP}{\vartheta_{+}(\zeta) < \gamma^{-1} \overline{\xi}}.
	\end{align}
	If $(\chi_{-},\chi_{+})\in \RPP^2$, then \eqref{e:derpsi} leads to
	\begin{align}
	&\psi'(\chi_{-})
	= \vartheta_{+}(\chi_{-})-
	\gamma^{-1} \overline{\xi}\label{e:psichimneg}\\
	& \psi'(\chi_{+}) =
	\chi_{+}\left(\gamma^{-1} \overline{\upsilon} - 
	\vartheta_{-}(\chi_{+})\right).
	\end{align}
	So, Conditions \ref{as:lepsi1} and \ref{as:lepsi2} are equivalent to
	\begin{align}
	&\vartheta_{+}(\chi_{-})-
	\gamma^{-1} \overline{\xi}<0\\
	& \chi_{+}\left(\gamma^{-1} \overline{\upsilon} - 
	\vartheta_{-}(\chi_{+})\right)> 0.
	\end{align}
	In view of \eqref{e:interchim} and \eqref{e:interchip}, these
	inequalities are satisfied if and only if $\chi_{-} < \chi_{+}$.
	This inequality is also obviously satisfied if $\chi_{-} = 0$ or $\chi_{+} = \pinf$.
	In addition, we have: 
	\begin{equation}
	\iflong(\forall \zeta \in \RPP)\qquad\fi \psi''(\zeta) = \gamma^{-1} \overline{\upsilon}-\vartheta_{-}(\zeta) + \zeta^{-1}(1+\zeta^{-2}) \varphi''(\zeta^{-1}).
	\end{equation}
	When $\zeta > \chi_{-}\neq \pinf$, $\gamma^{-1} \overline{\upsilon}-\vartheta_{-}(\zeta)> 0$, and the convexity of $\varphi$ yields
	$\psi''(\zeta) > 0$. This shows that $\psi$ is strictly convex on
	$]\chi_{-},\pinf[$. 
	
	If Conditions \ref{as:lepsi0}-\ref{as:lepsi2}
	are satisfied, due to the continuity of $\psi'$, there exists 
	$\widehat{\zeta}\in ]\chi_{-},\chi_{+}[$ such that $\psi'(\widehat{\zeta}) = 0$.
	Because of the
	strict convexity of $\psi$ on $]\chi_{-},\pinf[$,
	$\widehat{\zeta}$ is the unique minimizer
	of $\psi$ on this interval.
\end{proof}

The required assumptions in the previous lemma can often be simplified as stated below.

\begin{lemma}\label{le:psibis}
	Let $\gamma \in \RPP$ and $(\overline{\upsilon},\overline{\xi})\in
	\RR^2$.
	If $(\chi_{-},\chi_{+})\in \RPP^2$, then 
	Conditions \ref{as:lepsi1} and \ref{as:lepsi2}
	in Lemma~\ref{le:psi}
	are equivalent to: $\chi_{-} < \chi_{+}$.
	If $\chi_{-}\in \RPP$ and $\chi_{+} = \pinf$ (resp. $\chi_{-} = 0$ and
	$\chi_{+} \in \RPP$), Conditions \ref{as:lepsi1}-\ref{as:lepsi2} 
	are satisfied if and only if $\lim_{\zeta\to \pinf} \psi'(\zeta) > 0$
	(resp. $\lim_{\substack{\zeta \to 0\\ \zeta > 0}} \psi'(\zeta) < 0$).
\end{lemma}
\begin{proof}
	If $(\chi_{-},\chi_{+})\in \RPP^2$, we have already shown that
	Conditions \ref{as:lepsi1} and \ref{as:lepsi2} 
	are satisfied if and only $\chi_{-} < \chi_{+}$.
	
	If $\chi_{-} \in \RPP$ and
	$\chi_{+}=\pinf$ (resp. $\chi_{-} = 0$ and $\chi_{+} \in \RPP$), we still have
	\begin{align}
	&\psi'(\chi_{-})
	= \vartheta_{+}(\chi_{-})-
	\gamma^{-1} \overline{\xi} < 0\\
	(\text{resp.}\;\; & \psi'(\chi_{+}) =
	\chi_{+}\left(\gamma^{-1} \overline{\upsilon} - 
	\vartheta_{-}(\chi_{+})\right) > 0),
	\end{align}
	which shows that Condition \ref{as:lepsi1} (resp. Condition
	\ref{as:lepsi2}) is always satisfied.
\end{proof}

By using the same expressions of $\chi_{-}$ and $\chi_{+}$ as in the previous lemmas, we obtain the
following characterization of the proximity operator of any scaled
version of $\Phi$:
\begin{proposition}\label{p:psi}
	Let $\gamma \in \RPP$ and $(\overline{\upsilon},\overline{\xi})\in
	\RR^2$.
	$\prox_{\gamma \Phi}(\overline{\upsilon},\overline{\xi})\in
	\RPP^2$ if and only if Conditions~\ref{as:lepsi0}-\ref{as:lepsi2} in Lemma \ref{le:psi}
	are satisfied. When these conditions hold,
	\begin{equation}
	\prox_{\gamma \Phi}(\overline{\upsilon},\overline{\xi}) = 
	\big(\overline{\upsilon}-\gamma\,  \vartheta_{-}(\widehat{\zeta}),
	\overline{\xi}- \gamma\, \vartheta_{+}(\widehat{\zeta})\big)
	\label{e:thetap12}
	\end{equation} 
	where $\widehat{\zeta} < \chi_{+}$
	is the unique minimizer of $\psi$ on $]\chi_{-},\pinf[$.
\end{proposition}
\begin{proof}
	For every $(\overline{\upsilon},\overline{\xi})\in \RR^2$, such
	that Conditions~\ref{as:lepsi0}-\ref{as:lepsi2} in Lemma \ref{le:psi}
	hold,  let
	\begin{align}
	\upsilon = \overline{\upsilon}-\gamma\,
	\vartheta_{-}(\widehat{\zeta})
	\label{e:upsilonprox}\\
	\xi = \overline{\xi}- \gamma\,
	\vartheta_{+}(\widehat{\zeta})
	\label{e:chiprox}
	\end{align}
	where the existence of $\widehat{\zeta}\in
	]\chi_{-},\chi_{+}[$ is
	guaranteed by Lemma~\ref{le:psi}.
	As consequences of \eqref{e:interchim} and \eqref{e:interchip}, 
	$\upsilon$ and $\xi$ are positive. In addition, since
	\begin{equation}
	\psi'(\widehat{\zeta})=0 \quad
	\Leftrightarrow\quad \widehat{\zeta}
	\big(\gamma^{-1} \overline{\upsilon}-
	\vartheta_{-}(\widehat{\zeta})\big)
	= \gamma^{-1} \overline{\xi}-\vartheta_{+}(\widehat{\zeta})
	\end{equation}
	we derive
	from \eqref{e:upsilonprox} and \eqref{e:chiprox} that
	$\widehat{\zeta}= \xi/\upsilon > 0$.
	This allows us to re-express \eqref{e:upsilonprox} and
	\eqref{e:chiprox} as
	\begin{align}
	&\upsilon-\overline{\upsilon}+\gamma
	\varphi'\Big(\frac{\upsilon}{\xi}\Big)=0
	\label{e:thetap1}\\
	\nonumber\\[-0.5em]
	&\xi-\overline{\xi}+ \gamma
	\left(\varphi\Big(\frac{\upsilon}{\xi}\Big)-\frac{\upsilon}{\xi}\varphi'\Big(\frac{\upsilon}{\xi}\Big)\right)=0,
	\label{e:thetap2}
	\end{align} 
	that is
	\begin{align}
	&\upsilon-\overline{\upsilon}+\gamma \frac{\partial \Phi}{\partial\upsilon}(\upsilon,\xi) = 0 \\
	\nonumber\\[-0.5em]
	&\xi-\overline{\xi}+ \gamma \frac{\partial \Phi}{\partial\xi}(\upsilon,\xi) = 0. 
	\end{align} 
	The latter equations are satisfied if and only if \cite{Combettes_P_2010_inbook_proximal_smsp}
	\begin{equation}
	(\upsilon,\xi) = \prox_{\gamma\Phi}(\overline{\upsilon},\overline{\xi}).
	\end{equation}
	
	Conversely, for every $(\overline{\upsilon},\overline{\xi})\in \RR^2$, let
	$(\upsilon,\xi) = \prox_{\gamma
		\Phi}(\overline{\upsilon},\overline{\xi})$. If $(\upsilon,\xi)\in
	\RPP^2$, $(\upsilon,\xi)$ satisfies \eqref{e:thetap1} and
	\eqref{e:thetap2}. By setting $\widetilde{\zeta}= \xi/\upsilon > 0$,
	after simple calculations, we find
	\begin{align}
	&\upsilon = \overline{\upsilon}-\gamma\,
	\vartheta_{-}(\widetilde{\zeta})> 0\label{e:thetapp1}
	\\
	&\xi = \overline{\xi}- \gamma\, \vartheta_{+}(\widetilde{\zeta})>0
	\label{e:thetapp2}\\
	&\psi'(\widetilde{\zeta}) = 0 \label{e:thetapp3}.
	\end{align}
	According to \eqref{e:interchim} and \eqref{e:interchip}, 
	\eqref{e:thetapp1} and \eqref{e:thetapp2} imply that
	$\chi_{-}\neq \pinf$, $\chi_{+} \neq -\infty$, and
	$\widetilde{\zeta} \in
	]\chi_{-},\chi_{+}[$.
	In addition, according to Lemma~\ref{le:psi}, $\psi'$ is strictly
	increasing on $]\chi_{-},\pinf[$ (since $\psi$ is strictly convex on
	this interval). Hence,
	$\psi'$ has a  limit at $\chi_{-}$ (which may be equal to $-\infty$ when $\chi_{-} = -\infty$), and
	Condition~\ref{as:lepsi1} is satisfied. Similarly, $\psi'$ has a 
	limit  at $\chi_{+}$ (possibly equal to $\pinf$
	when $\chi_{+} = \pinf$), and
	Condition \ref{as:lepsi2} is satisfied.
\end{proof}

\begin{remark}\label{re:sym}
	In \eqref{e:perspf}, a special case arises when
	\begin{equation}\label{e:phisym}
	(\forall \zeta \in\RPP) \qquad \varphi(\zeta) = \widetilde{\varphi}(\zeta)+\zeta \widetilde{\varphi}(\zeta^{-1})
	\end{equation}
	where $\widetilde{\varphi}$ is a twice differentiable convex function on $\RPP$.
	Then $\Phi$ takes a symmetric form, leading to $\mathcal{L}$-divergences. It can then be deduced from \eqref{e:defvarthetap} that, for every $\zeta \in \RPP$, 
	\begin{equation}\label{e:varthetasym}
	\vartheta_-(\zeta) = \vartheta_+(\zeta^{-1})
	= \widetilde{\varphi}(\zeta)+\widetilde{\varphi}'(\zeta^{-1})- \zeta \widetilde{\varphi}'(\zeta).
	\end{equation}
\end{remark}

\section{Examples}
\label{s:examples}

\subsection{Kullback-Leibler divergence}\label{se:DKL}%

Let us now apply the results in the previous section to the function
\begin{equation}\label{e:divKL}
\Phi(\upsilon,\xi) =
\begin{cases}
\upsilon \ln\left(\dfrac{\upsilon}{\xi}\right) + \xi - \upsilon
& \mbox{if $(\upsilon,\xi) \in \RPP^2$}\\
\xi & \mbox{if $\upsilon = 0$ and $\xi \in \RP$}\\
\pinf & \mbox{otherwise.}
\end{cases}
\end{equation}
This is a function in $\Gamma_0(\RR^2)$ satisfying \eqref{e:perspf}
with 
\begin{equation}
(\forall \zeta \in \RPP)\qquad \varphi(\zeta) = \zeta \ln\zeta - \zeta
+ 1.
\end{equation}

\begin{proposition}\label{p:proxKL}
	The proximity operator of $\gamma \Phi$ with $\gamma \in \RPP$ is, for every $(\overline{\upsilon},\overline{\xi})\in \RR^2$,
	\begin{equation}\label{e:proxDxy_KLD} 
	\prox_{\gamma \Phi}(\overline{\upsilon},\overline{\xi}) = 
	\begin{cases}
	(\upsilon,\xi) & \mbox{if $\exp\left(\gamma^{-1}\overline{\upsilon}\right) >
		1-\gamma^{-1}\overline{\xi}$}\\
	(0,0) & \mbox{otherwise,}
	\end{cases}
	\end{equation}
	where
	\begin{align}
	\upsilon &= \overline{\upsilon} + \gamma \ln \widehat{\zeta}\label{eq:upsilon_KLD}\\
	\xi &= \overline{\xi}+\gamma \left(\widehat{\zeta}^{-1}-1\right)\label{eq:xi_KLD}
	\end{align}
	and $\widehat{\zeta}$ is the unique minimizer on
	$]\exp(-\gamma^{-1}\overline{\upsilon}),\pinf[$
	of
	\begin{align}\label{e:psiKL}
	\psi\colon &\RPP \to \RR\colon\\ 
	&\zeta  \mapsto
	\Big(\frac{\zeta^2}{2}-1\Big) \ln\zeta + \frac12 \Big(\gamma^{-1}\overline{\upsilon}-\frac12\Big)\zeta^2 
	+(1-\gamma^{-1} \overline{\xi}) \zeta.  \nonumber
	\end{align}
\end{proposition}
\begin{proof}
	For  every $(\overline{\upsilon},\overline{\xi})\in \RR^2$,
	$(\upsilon,\xi) = \prox_{\gamma
		\Phi}(\overline{\upsilon},\overline{\xi})$
	is such that $(\upsilon,\xi)\in \dom\Phi$ \cite{Moreau_JJ_1962_cras_Fonctions_cdeppdueh}.
	Let us first note that
	\begin{equation}\label{e:condupsilonuni}
	\upsilon \in \RPP \quad\Leftrightarrow\quad (\upsilon,\xi) \in \RPP^2.
	\end{equation}
	We are now able to apply Proposition \ref{p:psi},
	where $\psi$ is given by \eqref{e:psiKL} and, for every $\zeta \in \RPP$, 
	\begin{align}
	\Theta(\zeta) &= \frac{\zeta^2}{2}\left(\frac12 - \ln\zeta\right)-1\\
	\vartheta_{-}(\zeta) &= -\ln \zeta\\
	\vartheta_{+}(\zeta) &= 1 - \zeta^{-1}.
	\end{align}
	In addition,
	\begin{align}
	&\chi_{-} = \exp(-\gamma^{-1}\overline{\upsilon})\\
	&\chi_{+} = 
	\begin{cases}
	(1-\gamma^{-1}\overline{\xi})^{-1}& \mbox{if $\overline{\xi} < \gamma$}\\
	\pinf & \mbox{otherwise.}
	\end{cases}
	\end{align}
	According to \eqref{e:condupsilonuni} and Proposition \ref{p:psi},
	$\upsilon \in \RPP$ if and only if
	Conditions~\ref{as:lepsi0}-\ref{as:lepsi2} 
	in Lemma~\ref{le:psi} hold. Since $\chi_{-} \in \RPP$
	and $\lim_{\zeta\to \pinf} \psi'(\zeta) = \pinf$, 
	Lemma \ref{le:psibis} shows that these conditions are satisfied if and
	only if
	\begin{equation}
	\overline{\xi} \ge \gamma \quad\textrm{or}\quad \big(\overline{\xi} < \gamma \quad\text{and}\quad \exp(-\overline{\upsilon}/\gamma) < (1-\gamma^{-1}\overline{\xi})^{-1}\big),
	\end{equation}
	which is equivalent to
	\begin{equation}\label{e:ineqproxKL}
	\exp(\overline{\upsilon}/\gamma) > 1-\gamma^{-1}\overline{\xi}.
	\end{equation}
	Under this assumption, Proposition \ref{p:psi} leads to the
	expressions \eqref{eq:upsilon_KLD} and \eqref{eq:xi_KLD} of the proximity operator, where $\widehat{\zeta}$ is the unique minimizer on $]\exp(-\overline{\upsilon}/\gamma),\pinf[$ of the function $\psi$.
	
	We have shown that
	$\upsilon > 0 \Leftrightarrow$ \eqref{e:ineqproxKL}. So, $\upsilon = 0$ when \eqref{e:ineqproxKL} is not satisfied.
	Then, the expression of $\xi$ simply reduces to
	the asymmetric soft-thresholding rule \cite{Combettes2007}:
	\begin{equation}\label{e:threshxi}
	\xi = \begin{cases}
	\overline{\xi} - \gamma & \mbox{if $\overline{\xi} > \gamma$}\\
	0 & \mbox{otherwise.}
	\end{cases}
	\end{equation}
	However, $\exp(\gamma^{-1}\overline{\upsilon})  \le
	1-\gamma^{-1}\overline{\xi} \Rightarrow 
	\overline{\xi}<\gamma $, so that $\xi$ is necessarily equal to $0$.
\end{proof}

\begin{remark}
	More generally, we can derive the proximity operator of
	\begin{equation}
	\widetilde{\Phi}(\upsilon,\xi) =
	\begin{cases}
	\upsilon \ln\left(\dfrac{\upsilon}{\xi}\right) \!+\! \kappa(\xi - \upsilon)
	& \mbox{if $(\upsilon,\xi) \in \RPP^2$}\\
	\kappa \xi & \mbox{if $\upsilon = 0$ and $\xi \in \RP$}\\
	\pinf & \mbox{otherwise,}
	\end{cases}
	\end{equation}
	where $\kappa \in \RR$. Of particular interest in the literature is the case when
	$\kappa = 0$ \cite{Blahut72,Arimoto1972,Dupe_FX_2008_ip_proximal_ifdpniusr,Pustelnik_N_2011_tip_PPXA}. From Proposition~\ref{p:propprox}\ref{p:propproxiv}, we get, for every $\gamma \in \RPP$ and for every $(\overline{\upsilon},\overline{\xi})\in \RR^2$,
	\begin{equation}
	\prox_{\gamma \widetilde{\Phi}}(\overline{\upsilon},\overline{\xi})=
	\prox_{\gamma \Phi}(\overline{\upsilon}+\gamma\kappa-\gamma,\overline{\xi}-\gamma\kappa+\gamma),
	\end{equation}
	where $\prox_{\gamma\Phi}$ is provided by Proposition \ref{p:proxKL}.
\end{remark}

\begin{remark}
	It can be noticed that
	\begin{equation}
	\psi'(\widehat{\zeta}) = \widehat{\zeta} \ln \widehat{\zeta} + \gamma^{-1}\overline{\upsilon}\widehat{\zeta} - {\widehat{\zeta}}^{-1} +1-\gamma^{-1}\overline{\xi} = 0
	\end{equation}
	is equivalent to
	\begin{equation}\label{e:alphaKLWspec}
	\widehat{\zeta}^{-1} \exp\Big(\widehat{\zeta}^{-1}\big(\widehat{\zeta}^{-1} + \gamma^{-1}\overline{\xi}-1\big)\Big) = \exp(\gamma^{-1}\overline{\upsilon}).
	\end{equation}
	In the case where $\overline{\xi}=\gamma$, the above equation reduces to
	\begin{align}
	& 2\widehat{\zeta}^{-2} \exp\big(2\widehat{\zeta}^{-2}\big) = 2\exp(2\gamma^{-1}\overline{\upsilon})\nonumber\\
	\Leftrightarrow \qquad & \widehat{\zeta} =  \left(\frac{2}{W(2e^{2\gamma^{-1}\overline{\upsilon}})}\right)^{1/2}
	\end{align}
	where $W$ is the Lambert W function \cite{Corless96onthe}.
	When $\overline{\xi}\neq\gamma$, although a closed-form expression of \eqref{e:alphaKLWspec} is not available, efficient numerical methods to compute $\widehat{\zeta}$ can be developed.
\end{remark}

\begin{remark}
	To minimize $\psi$ in \eqref{e:psiKL}, we need to find the zero on $]\exp(-\gamma^{-1}\overline{\upsilon}),+\infty[$ of the function: $(\forall \zeta\in \RPP)$
	\begin{equation}\label{e:derpsiKL}
	\psi'(\zeta)  = 
	{\zeta} \ln {\zeta} + \gamma^{-1}\overline{\upsilon}\,{\zeta} - {\zeta}^{-1}  +1-\gamma^{-1}\overline{\xi}.
	\end{equation}
	This can be performed by Algorithm \ref{algo:Newton_DKL}, the convergence of which is proved in Appendix~\ref{s:algoNewton_DKL}. 
	
	\begin{minipage}{0.9\linewidth}
		\begin{algorithm}[H]
			\caption{Newton method for minimizing \eqref{e:psiKL}.}\label{algo:Newton_DKL}
			\small
			\vspace{0.5em}
			\textsc{Set} $\widehat{\zeta}_0 = \exp(-\gamma^{-1}\overline{\upsilon})$\\[0.5em]
			\noindent \textsc{For}\; $n = 0, 1, \dots$\\[-0.5em]
			\[
			\!\!\!\!\!\!\!\!\!\!\!\!\!\!\!\!\!\!\!\!\!\!\!\!\!\!\!\!\!\!
			\left\lfloor
			\begin{aligned}
			&\textrm{	$\widehat{\zeta}_{n+1}=\widehat{\zeta}_n-\psi'(\widehat{\zeta}_n)/\psi''(\widehat{\zeta}_n).$}\\
			\end{aligned}
			\right.
			\]
		\end{algorithm}
		\vspace{0.1cm}
	\end{minipage}
	
\end{remark}


\subsection{Jeffreys divergence} \label{se:DJK}

Let us now consider the symmetric form of \eqref{e:divKL} given by
\begin{equation}
\Phi(\upsilon,\xi)= \begin{cases}
(\upsilon-\xi) \big(\ln \upsilon - \ln \xi)
& \mbox{if $(\upsilon,\xi) \in \RPP^2$}\\
0 & \mbox{if $\upsilon = \xi = 0$}\\
\pinf & \mbox{otherwise.}
\end{cases}
\end{equation}
This function belongs to $\Gamma_0(\RR^2)$ and satisfies \eqref{e:perspf} and \eqref{e:phisym}
with 
\begin{equation}
(\forall \zeta \in \RPP)\qquad \widetilde{\varphi}(\zeta) = -\ln\zeta.
\end{equation}

\begin{proposition}\label{p:proxJKL}
	The proximity operator of $\gamma \Phi$ with $\gamma \in \RPP$ is, for every $(\overline{\upsilon},\overline{\xi})\in \RR^2$,
	\begin{equation}\label{e:proxDxy_JKL} 
	\prox_{\gamma \Phi}(\overline{\upsilon},\overline{\xi}) = 
	\begin{cases}
	(\upsilon,\xi) &\mbox{if $W(e^{1-\gamma^{-1} \overline{\upsilon}})W(e^{1 - \gamma^{-1} \overline{\xi}}) < 1$}\\
	(0,0) & \mbox{otherwise}
	\end{cases}
	\end{equation}
	where
	\begin{align}
	\upsilon &= \overline{\upsilon} + \gamma \big(\ln \widehat{\zeta} + \widehat{\zeta}-1)\\
	\xi &= \overline{\xi}-\gamma \big(\ln \widehat{\zeta} -\widehat{\zeta}^{-1}+1)
	\end{align}
	and $\widehat{\zeta}$ is the unique minimizer on
	$]W(e^{1-\gamma^{-1}\overline{\upsilon}}),\pinf[$
	of
	\begin{align}\label{e:psiJK}
	\psi&\colon \RPP \to \RR\colon \nonumber \\
	& \zeta \mapsto
	\Big(\frac{\zeta^2}{2}+\zeta-1\Big) \ln\zeta + \frac{\zeta^3}{3}
	+ \frac12 \Big(\gamma^{-1}\overline{\upsilon}- \frac{3}{2}\Big)
	\zeta^2 -\gamma^{-1} \overline{\xi} \zeta.
	\end{align}
\end{proposition}
\iflong
\begin{proof}
	We apply Proposition \ref{p:psi}
	where $\psi$ is given by \eqref{e:psiJK} and, for every $\zeta \in \RPP$, 
	\begin{align}
	&\Theta(\zeta) = \zeta^2 \Big(\frac{3}{4} - \frac{\zeta}{3} - \frac12 \ln \zeta \Big)\\
	&\vartheta_{-}(\zeta) = \vartheta_{+}(\zeta^{-1}) = -\ln \zeta - \zeta +1.
	\label{e:varthetajef}
	\end{align}
	The above equalities have been derived from \eqref{e:phisym} and \eqref{e:varthetasym}.
	It can be deduced from \eqref{e:defchim}, \eqref{e:defchip} and \eqref{e:varthetajef} that
	\begin{align}
	& \chi_- + \ln \chi_- = 1-\gamma^{-1} \overline{\upsilon}\\
	& \chi_+^{-1} + \ln(\chi_+^{-1}) = 1 - \gamma^{-1} \overline{\xi}
	\end{align}
	that is
	\begin{align}
	& \chi_- = W(e^{1-\gamma^{-1} \overline{\upsilon}})\label{eq:chi_minus}\\
	& \chi_+ = \big(W(e^{1 - \gamma^{-1} \overline{\xi}})\big)^{-1}\label{eq:chi_plus}.
	\end{align}
	
	According to Proposition \ref{p:psi},
	$\prox_{\gamma \Phi}(\overline{\upsilon},\overline{\xi})\in \RPP^2$ if and only if
	Conditions~\ref{as:lepsi0}-\ref{as:lepsi2} 
	in Lemma~\ref{le:psi} hold.  
	Lemma~\ref{le:psibis} shows that these conditions are satisfied if and
	only if
	\begin{align}\label{e:ineqproxJK}
	W(e^{1-\gamma^{-1} \overline{\upsilon}})W(e^{1 - \gamma^{-1} \overline{\xi}}) < 1.
	\end{align}
	Under this assumption, the expression of the proximity operator follows from
	Proposition \ref{p:psi} and \eqref{e:varthetajef}.
	
	\noindent We have shown that $\prox_{\gamma \Phi}(\overline{\upsilon},\overline{\xi})\in \RPP^2 \Leftrightarrow$ 
	\eqref{e:ineqproxJK}. Since  $\prox_{\gamma \Phi}(\overline{\upsilon},\overline{\xi}) \in \dom \Phi$, we necessarily get $\prox_{\gamma \Phi}(\overline{\upsilon},\overline{\xi}) = (0,0)$, when \eqref{e:ineqproxJK} is not satisfied.
\end{proof}
\fi

\begin{remark}
To minimize $\psi$ in \eqref{e:psiJK}, we need to find the zero on $[\chi_{-}, \chi_{+}]$ of the function: $(\forall \zeta\in \RPP)$
	\begin{equation}\label{e:derJK}
	\psi'(\zeta)  = (\zeta+1) \ln\zeta + \frac{\zeta}{2} - \zeta^{-1} + \zeta^2 + \Big(\gamma^{-1}\overline{\upsilon}- \frac{3}{2}\Big) \zeta + 1 - \gamma^{-1} \overline{\xi}.
	\end{equation}
	This can be performed by 
	\iflong
	Algorithm \ref{algo:Newton}, the convergence of which is proved in Appendix~\ref{s:algoNewton}.
	
	\begin{minipage}{0.9\linewidth}
		\begin{algorithm}[H]
			\caption{Projected Newton for minimizing \eqref{e:psiJK}.}
			\label{algo:Newton}
			\small
			\vspace{0.5em}
			\textsc{Set} $\widehat{\zeta}_0 \in [\chi_-,\chi_+]$ \quad (see \eqref{eq:chi_minus}--\eqref{eq:chi_plus} for the bound expressions)\\[0.5em]
			\noindent \textsc{For}\; $n = 0, 1, \dots$\\[-0.5em]
			\[
			\left\lfloor
			\begin{aligned}
			&\textrm{	$\widehat{\zeta}_{n+1}=\mathsf{P}_{[\chi_-,\chi_+]}\Big(\widehat{\zeta}_n-\psi'(\widehat{\zeta}_n)/\psi''(\widehat{\zeta}_n)\Big).$}\\
			\end{aligned}
			\right.
			\]
		\end{algorithm}
		\vspace{0.1cm}
	\end{minipage}
	
	\fi

\end{remark}

\begin{remark}
	From a numerical standpoint, to avoid the arithmetic overflow in the exponentiations when $\gamma^{-1} \overline{\upsilon}$ or $\gamma^{-1} \overline{\xi}$ tend to $-\infty$, one can use the asymptotic approximation of the Lambert W function for large values: for every $\tau \in \left[1,+\infty\right[$,
	\begin{equation}
	\tau - \ln \tau + \frac{1}{2} \frac{\ln \tau}{\tau} \le W\big(e^\tau\big) \le
	\tau - \ln \tau + \frac{e}{e-1}\frac{\ln \tau}{\tau},
	\end{equation} 
	with equality only if $\tau=1$ \cite{Hoorfar_inequalities}.
\end{remark}

\subsection{Hellinger divergence}\label{se:DH}
Let us now consider the function of $\Gamma_0(\RR^2)$ given by
\begin{equation}
\Phi(\upsilon,\xi)= \begin{cases}
(\sqrt{\upsilon}-\sqrt{\xi})^2
& \mbox{if $(\upsilon,\xi) \in \RP^2$}\\
\pinf & \mbox{otherwise.}
\end{cases}
\end{equation}
This symmetric function satisfies \eqref{e:perspf} and \eqref{e:phisym} with
\begin{equation}
(\forall \zeta \in \RPP) \qquad \widetilde{\varphi}(\zeta) = \zeta - \sqrt{\zeta}.
\end{equation}

\begin{proposition}\label{p:hell}
	The proximity operator of $\gamma \Phi$ with $\gamma \in \RPP$ is, for every $(\overline{\upsilon},\overline{\xi})\in \RR^2$, 
	\begin{equation}\label{e:proxHel} 
	\prox_{\gamma \Phi}(\overline{\upsilon},\overline{\xi}) = 
	\begin{cases}
	(\upsilon,\xi) &\mbox{if $\overline{\upsilon} \ge \gamma$ or $\left(1-\frac{\overline{\upsilon}}{\gamma}\right) \left(1-\frac{\overline{\xi}}{\gamma}\right) < 1$}\\
	(0,0) & \mbox{otherwise,}
	\end{cases}
	\end{equation}
	where
	\begin{align}
	\upsilon &= \overline{\upsilon} + \gamma ({\rho}-1) \\
	\xi &= \overline{\xi}+\gamma\left({\rho}^{-1}-1\right)
	\end{align}
	and ${\rho}$ is the unique solution on $\left]\max(1-\gamma^{-1}\overline{\upsilon},0),\pinf\right[$ of 
	\begin{equation}\label{e:3rdHel}
	\rho^{4}+ \left(\gamma^{-1}\overline{\upsilon}-1\right)\rho^3 + \left(1-\gamma^{-1}\overline{\xi}\right)\rho-1= 0.
	\end{equation}
\end{proposition}
\iflong
\vspace{0.5em}
\begin{proof}
	For  every $(\overline{\upsilon},\overline{\xi})\in \RR^2$,
	$(\upsilon,\xi) = \prox_{\gamma
		\Phi}(\overline{\upsilon},\overline{\xi})$
	is such that $(\upsilon,\xi)\in \RP^2$ \cite{Moreau_JJ_1962_cras_Fonctions_cdeppdueh}.
	By using the notation of Proposition \ref{p:psi} and by using Remark \ref{re:sym}, we have that, for every $\zeta \in \RPP$,
	\begin{align}
	&\Theta(\zeta) = \frac{\zeta^2}{2} - \frac{2}{5} \zeta^{5/2}+1\\
	&\vartheta_{-}(\zeta) = \vartheta_{+}(\zeta^{-1}) = 1-\sqrt{\zeta}
	\end{align}
	and
	\begin{align}
	&\chi_{-} = 
	\begin{cases}
	(1-\gamma^{-1}\overline{\upsilon})^2
	& \mbox{if $\overline{\upsilon} < \gamma$}\\
	0 & \mbox{otherwise}
	\end{cases}
	\\\nonumber\\[-0.5em]
	&\chi_{+} = 
	\begin{cases}
	(1-\gamma^{-1}\overline{\xi})^{-2}& \mbox{if $\overline{\xi} < \gamma$}\\
	\pinf & \mbox{otherwise.}
	\end{cases}
	\end{align}
	According to Proposition \ref{p:psi},
	$(\upsilon,\xi) \in \RPP^2$ if and only if Conditions~\ref{as:lepsi0}-\ref{as:lepsi2} in Lemma \ref{le:psi} hold. Under these conditions, Proposition \ref{p:psi} leads to
	\begin{align}
	&\upsilon = \overline{\upsilon} + \gamma (\widehat{\zeta}^{1/2}-1)\\
	&\xi = \overline{\xi}+\gamma\left(\widehat{\zeta}^{-1/2}-1\right) 
	\end{align}
	where $\widehat{\zeta}$ is the unique minimizer on
	$]\chi_{-},\pinf[$ of the function defined as, for every $\zeta \in \RPP$, 
	\begin{equation}\label{e:psiHel}
	\psi(\zeta) =
	\frac{2}{5} \zeta^{5/2} - 2\zeta^{1/2} +
	\frac{\gamma^{-1}\overline{\upsilon}-1}{2}
	\zeta^2 + (1-\gamma^{-1}\overline{\xi})\zeta.
	\end{equation}
	This means that $\widehat{\zeta}$ is the unique solution on
	$]\chi_{-},\pinf[$ of the equation:
	\begin{equation}
	\psi'(\widehat{\zeta}) = \widehat{\zeta}^{3/2} - \widehat{\zeta}^{-1/2} +
	(\gamma^{-1}\overline{\upsilon}-1)
	\widehat{\zeta} + 1-\gamma^{-1}\overline{\xi}= 0.
	\end{equation}
	By setting $\rho = \widehat{\zeta}^{1/2}$, \eqref{e:3rdHel} is obtained.

	Since 
	$\lim_{\substack{\zeta \to 0\\ \zeta > 0}} \psi'(\zeta) = -\infty$
	and $\lim_{\zeta\to \pinf} \psi'(\zeta) = \pinf$, 
	Lemma \ref{le:psibis} shows that
	Conditions~\ref{as:lepsi0}-\ref{as:lepsi2} are satisfied if and only if
	\begin{align}\label{e:ineqHelbefore}
	&\overline{\upsilon} < \gamma,\;\; \overline{\xi} < \gamma, \;\;\text{and}\;\;
	(1-\gamma^{-1} \overline{\upsilon})^2
	< (1-\gamma^{-1}\overline{\xi})^{-2}\nonumber\\
	\text{or}\;\;& \overline{\upsilon} < \gamma\;\; \text{and}\;\; \overline{\xi}
	\ge \gamma\nonumber\\
	\text{or}\;\;& \overline{\upsilon} \ge \gamma \;\; \text{and}\;\; \overline{\xi}
	< \gamma\nonumber\\
	\text{or}\;\;& \overline{\upsilon} \ge \gamma \;\; \text{and}\;\; \overline{\xi}
	\ge \gamma
	\end{align}
	or, equivalently
	\begin{align}\label{e:ineqHel}
	&\overline{\upsilon} < \gamma\;\;\text{and}\;\;
	(1-\gamma^{-1}\overline{\upsilon}) 
	(1-\gamma^{-1}\overline{\xi}) < 1\nonumber\\
	\text{or}\;\;& \overline{\upsilon} \ge \gamma.
	\end{align}

	In turn, when \eqref{e:ineqHel} is not satisfied, we necessarily have $\upsilon
	= 0$ or $\xi = 0$. In the first case, the expression of 
	$\xi$ is simply given by the asymmetric soft-thresholding rule
	in \eqref{e:threshxi}.
	Similarly, in the second case, we have
	\begin{equation}\label{e:upsiloncp}
	\upsilon = \begin{cases}
	\overline{\upsilon} - \gamma & \mbox{if $\overline{\upsilon} > \gamma$}\\
	0 & \mbox{otherwise.}
	\end{cases}
	\end{equation}
	However, when $\overline{\upsilon} > \gamma$ or $\overline{\xi} >
	\gamma$, \eqref{e:ineqHelbefore} is always satisfied, so that 
	$\upsilon = \xi = 0$.
	
	Altogether, the above results yield the expression of the proximity
	operator in \eqref{e:proxHel}.
\end{proof}
\fi

\subsection{Chi square divergence} \label{se:DCS}

Let us now consider the function of $\Gamma_0(\RR^2)$ given by
\begin{equation}
\Phi(\upsilon,\xi)=
\begin{cases}
\displaystyle \frac{(\upsilon-\xi)^2}{\xi}
& \mbox{if $\upsilon\in \RP$ and $\xi \in \RPP$}\\
0 & \mbox{if $\upsilon = \xi = 0$}\\
\pinf & \mbox{otherwise.}
\end{cases}
\end{equation}
This function satisfies \eqref{e:perspf} with 
\begin{equation}
(\forall \zeta \in \RPP)\qquad \varphi(\zeta) = (\zeta-1)^2.
\end{equation}

\newcommand{\chisquarecond}{$\overline{\xi}> - \left(\overline{\upsilon}+\dfrac{\overline{\upsilon}^2}{4\gamma}\right)$}

\begin{proposition}\label{p:chisquare}
	The proximity operator of $\gamma \Phi$ with $\gamma \in \RPP$ is, for every $(\overline{\upsilon},\overline{\xi})\in \RR^2$,
	\begin{equation}\label{e:proxChi} 
	\prox_{\gamma \Phi}(\overline{\upsilon},\overline{\xi}) = 
	\begin{cases}
	(\upsilon,\xi) &\mbox{if $\overline{\upsilon} > -2\gamma$ and \ \iflong\chisquarecond\fi} \\
	\big(0,\max\{\overline{\xi}-\gamma,0\}\big) & \mbox{otherwise},
	\end{cases}
	\end{equation}
	where
	\begin{align}
	\upsilon &= \overline{\upsilon} + 2 \gamma(1-{\rho})\\
	\xi &= \overline{\xi}+\gamma({\rho}^2-1) 
	\end{align}
	and ${\rho}$ is the unique solution on $]0,1+\gamma^{-1}{\overline{\upsilon}}/{2}[$ of
	\begin{equation}\label{e:3rdchi}
	\rho^3+\left(1+\gamma^{-1}\overline{\xi}\right) \rho - \gamma^{-1}\overline{\upsilon} -2 = 0.
	\end{equation}
\end{proposition}
\iflong
\begin{proof}
	By proceeding similarly to the proof of Proposition~\ref{p:hell}, we have that, for every $\zeta \in \RPP$, $
	\Theta(\zeta) = 2\zeta-\zeta^2$, $\vartheta_{-}(\zeta) = 2(\zeta ^{-1}-1)$, $\vartheta_{+}(\zeta) =  1 - \zeta^{-2}$, and
	\begin{align}
	&\chi_{-} = 
	\begin{cases}
	2\big(2+\gamma^{-1}\overline{\upsilon}\big)^{-1}
	& \mbox{if $\overline{\upsilon} > -2\gamma$}\\
	\pinf & \mbox{otherwise}
	\end{cases}
	\\
	&\chi_{+} = 
	\begin{cases}
	\big(1-\gamma^{-1} \overline{\xi}\big)^{-1/2}
	& \mbox{if $\overline{\xi} < \gamma$}\\
	\pinf & \mbox{otherwise.}
	\end{cases}
	\end{align}
	According to Proposition \ref{p:psi},
	$(\upsilon,\xi) \in \RPP^2$ if and only if
	Conditions~\ref{as:lepsi0}-\ref{as:lepsi2} in Lemma \ref{le:psi}
	hold. Then,
	$(\upsilon,\xi) = \prox_{\gamma
		\Phi}(\overline{\upsilon},\overline{\xi})$ is such that
	$\upsilon = \overline{\upsilon} + 2 \gamma(1-\widehat{\zeta}^{-1})$ and $\xi = \overline{\xi}+\gamma(\widehat{\zeta}^{-2}-1)$, 
	where $\widehat{\zeta}$ is the unique minimizer on
	$]\chi_{-},\pinf[$ of the function:
	\begin{equation*}\label{e:psiChi}
	\psi\colon \RPP \to \RR\colon \zeta \mapsto
	\Big(1+\frac{\gamma^{-1}\overline{\upsilon}}{2}\Big) \zeta^2 - (1+\gamma^{-1}\overline{\xi})\zeta -2 + \zeta^{-1}.
	\end{equation*}
	Thus, $\widehat{\zeta}$ is the unique solution on
	$]\chi_{-},\pinf[$ of the equation:
	\begin{equation}
	\psi'(\widehat{\zeta}) = (2+\gamma^{-1}\overline{\upsilon}) \widehat{\zeta} - 1-\gamma^{-1}\overline{\xi} -\widehat{\zeta}^{-2} = 0.
	\end{equation}
	By setting $\rho = \widehat{\zeta}^{-1}$, \eqref{e:3rdchi} is obtained. Lemma \ref{le:psibis} shows that Conditions~\ref{as:lepsi1} and
	\ref{as:lepsi2} are satisfied if and only if
	\begin{align}
	&\overline{\upsilon} > -2\gamma,\;\; \overline{\xi} < \gamma, \;\;\text{and}\;\;
	\frac{2}{2+\gamma^{-1}\overline{\upsilon}} <
	\frac{1}{\sqrt{1-\gamma^{-1} \overline{\xi}}} 
	\nonumber\\
	\text{or}\;\;& \overline{\upsilon} > -2\gamma\;\; \text{and}\;\; 
	\overline{\xi} \ge \gamma
	\end{align}
	or, equivalently,
	\begin{align}\label{e:ineqchi}
	&\overline{\upsilon} > -2\gamma,\;\; \overline{\xi} < \gamma,  
	\;\;\text{and}\;\;
	1-{\gamma}^{-1}{\overline{\xi}} <
	\left(1+{(2\gamma)}^{-1}{\overline{\upsilon}}\right)^2
	\nonumber\\
	\text{or}\;\;& \overline{\upsilon} > -2\gamma\;\; \text{and}\;\; 
	\overline{\xi} \ge \gamma.
	\end{align}

	When \eqref{e:ineqchi} does not hold, we necessarily have $\upsilon
	= 0$. The end of the proof is similar to that of Proposition \ref{p:hell}.
\end{proof}
\fi

\subsection{Renyi divergence} \label{se:DR}

Let $\alpha \in ]1,\pinf[$ and consider the below function of $\Gamma_0(\RR^2)$
\begin{equation}
\Phi(\upsilon,\xi)=
\left\{
\begin{aligned}
&\dfrac{{\upsilon}^\alpha}{\xi^{\alpha-1}}
&& \mbox{if $\upsilon\in \RP$ and $\xi \in \RPP$}\\
&0 && \mbox{if $\upsilon = \xi = 0$}\\
&\pinf && \mbox{otherwise,}
\end{aligned}
\right.
\end{equation}
which corresponds to the case when
\begin{equation}
(\forall \zeta \in \RPP)\qquad \varphi(\zeta) = \zeta^\alpha.
\end{equation}
Note that the above function $\Phi$ allows us to generate the R\'enyi divergence
up to a log transform and a multiplicative constant.

\newcommand{\renyicond}{$\dfrac{\gamma^{\frac{1}{\alpha-1}}\overline{\xi}}{1-\alpha} <
	\left(\dfrac{\overline{\upsilon}}{\alpha}\right)^{\frac{\alpha}{\alpha-1}}$}

\begin{proposition}
	The proximity operator of $\gamma \Phi$ with $\gamma \in \RPP$ is, for every $(\overline{\upsilon},\overline{\xi})\in \RR^2$,
	\begin{equation}\label{e:proxReni}
	\prox_{\gamma \Phi}(\overline{\upsilon},\overline{\xi}) =
	\left\{
	\begin{aligned}
	&(\upsilon,\xi) &&\mbox{if $\overline{\upsilon}>0$ and \ \iflong\renyicond\fi} \\
	&\big(0,\max\{\overline{\xi},0\}\big) && \mbox{otherwise},
	\end{aligned}
	\right.
	\end{equation}
	where
	\begin{align}
	\upsilon &= \overline{\upsilon} - \gamma \alpha \widehat{\zeta}^{1-\alpha}
	\\
	\xi &= \overline{\xi}+ \gamma (\alpha-1) \widehat{\zeta}^{-\alpha}
	\end{align}
	and $\widehat{\zeta}$ is the unique solution on
	$](\alpha \gamma
	\overline{\upsilon}^{-1})^{\frac{1}{\alpha-1}},\pinf[$ 
	of 
	\begin{equation}\label{e:opteqRen}
	\gamma^{-1}\overline{\upsilon}\; \widehat{\zeta}^{1+\alpha} -
	\gamma^{-1}\overline{\xi}\, \widehat{\zeta}^{\alpha}
	- \alpha \widehat{\zeta}^{2} + 1-\alpha = 0.
	\end{equation}
\end{proposition}
\iflong
\begin{proof}
	We proceed similarly to the previous examples by noticing that, for every $\zeta \in \RPP$,
	\begin{align}
	\Theta(\zeta) &= \begin{cases}
	\frac{\alpha}{3-\alpha} \zeta^{3-\alpha} & \mbox{if $\alpha \neq 3$}\\
	\alpha \ln \zeta & \mbox{if $\alpha = 3$}
	\end{cases}\\
	\vartheta_{-}(\zeta)& = \alpha \zeta^{1-\alpha}\\
	\vartheta_{+}(\zeta) &= (1-\alpha) \zeta^{-\alpha}\\
	\psi'(\zeta) & = (1-\alpha) \zeta^{-\alpha} - \alpha \zeta^{2-\alpha}
	+ \gamma^{-1} \overline{\upsilon}\zeta - \gamma^{-1} \overline{\xi}
	\end{align}
	and
	\begin{align}
	&\chi_{-} = 
	\begin{cases}
	\displaystyle 
	\Big(\frac{\gamma \alpha}{\overline{\upsilon}}\Big)^{\frac{1}{\alpha-1}}
	& \mbox{if $\overline{\upsilon} > 0$}\\
	\pinf & \mbox{otherwise}
	\end{cases}
	\\
	&\chi_{+} = 
	\begin{cases}
	\displaystyle\Big(\frac{\gamma(1-\alpha)}{\overline{\xi}}\Big)^{1/\alpha}
	& \mbox{if $\overline{\xi} < 0$}\\
	\pinf & \mbox{otherwise.}
	\end{cases}
	\end{align}
\end{proof}
\fi

Note that \eqref{e:opteqRen} becomes a polynomial equation when $\alpha$ is
a rational number. In particular, when $\alpha = 2$, it reduces to
the cubic equation:
\begin{equation}
\rho^3+\left(2+\gamma^{-1}\overline{\xi}\right)\rho - \gamma^{-1} \overline{\upsilon} = 0
\end{equation}
with $\widehat{\zeta} = \rho^{-1}$.

\subsection{$\mathrm{I}_\alpha$ divergence} \label{se:DIa}

Let $\alpha \in ]0,1[$ and consider the function of $\Gamma_0(\RR^2)$ given by
\begin{equation}
\Phi(\upsilon,\xi)=
\left\{
\begin{aligned}
&\alpha \upsilon+(1-\alpha) \xi - \upsilon^\alpha\xi^{1-\alpha}
&& \mbox{if $(\upsilon,\xi)\in \RP^2$}\\
&\pinf && \mbox{otherwise}
\end{aligned}
\right.
\end{equation}
which corresponds to the case when
\begin{equation}
(\forall \zeta \in \RPP)\qquad \varphi(\zeta) = 1-\alpha+\alpha \zeta
-\zeta^\alpha.
\end{equation}

\newcommand{\Iacond}{$\Big.1-\dfrac{\overline{\xi}}{\gamma(1-\alpha)}<\left(1-\dfrac{\overline{\upsilon}}{\gamma\alpha}\right)^{\frac{\alpha}{\alpha-1}}$}

\begin{proposition}\label{prop:proxIalpha}
	The proximity operator of $\gamma \Phi$ with $\gamma \in \RPP$ is, for every $(\overline{\upsilon},\overline{\xi})\in \RR^2$,
	\begin{equation}\label{e:proxIalpha} 
	\prox_{\gamma \Phi}(\overline{\upsilon},\overline{\xi}) = 
	\left\{\begin{aligned}
	&(\upsilon,\xi) &&\mbox{if $\overline{\upsilon} \ge \gamma\alpha$ or \ \iflong\Iacond\fi}\\
	&(0,0) && \mbox{otherwise},
	\end{aligned}
	\right.
	\end{equation}
	where
	\begin{align}
	\upsilon &= \overline{\upsilon} +\gamma \alpha(\widehat{\zeta}^{1-\alpha}-1) \\
	\xi &= \overline{\xi}+ \gamma (1-\alpha) (\widehat{\zeta}^{-\alpha}-1)
	\end{align}
	and $\widehat{\zeta}$ is the unique solution on
	{\footnotesize{$\big]\big(\max\{1- \frac{\overline{\upsilon}}{\gamma\alpha},0\}\big)^{\frac{1}{1-\alpha}},\pinf\big[$}}
	of
	\begin{equation}\label{e:opteqIalpha}
	\alpha \widehat{\zeta}^2+
	(\gamma^{-1}\overline{\upsilon}-\alpha) \widehat{\zeta}^{\alpha+1}
	+(1-\alpha-\gamma^{-1}\overline{\xi})\widehat{\zeta}^\alpha = 1-\alpha.
	\end{equation}
\end{proposition}
\iflong
\begin{proof}
	We have then, for every $\zeta \in \RPP$,
	\begin{align}
	&\Theta(\zeta) = \alpha \Big(\frac{\zeta^2}{2} - \frac{\zeta^{3-\alpha}}{3-\alpha}\Big)\\
	&\vartheta_{-}(\zeta) = \alpha(1- \zeta^{1-\alpha})\\
	&\vartheta_{+}(\zeta) = (1-\alpha) (1-\zeta^{-\alpha})\\
	& \psi'(\zeta) = \alpha \zeta^{2-\alpha} + 
	\Big(\frac{\overline{\upsilon}}{\gamma}-\alpha\Big)\zeta +\frac{\alpha-1}{\zeta^{\alpha}} + 1 - \alpha - \frac{\overline{\xi}}{\gamma}
	\end{align}
	and
	\begin{align}
	&\chi_{-} = 
	\begin{cases}
	\displaystyle 
	\Big(1- \frac{\overline{\upsilon}}{\gamma \alpha}\Big)^{\frac{1}{1-\alpha}}
	& \mbox{if $\overline{\upsilon} < \gamma \alpha$}\\
	0 & \mbox{otherwise}
	\end{cases}
	\\
	&\chi_{+} = 
	\begin{cases}
	\displaystyle\Big(1-\frac{\overline{\xi}}{\gamma(1-\alpha)}\Big)^{-1/\alpha}
	& \mbox{if $\overline{\xi} < \gamma (1-\alpha)$}\\
	\pinf & \mbox{otherwise.}
	\end{cases}
	\end{align}
	The result follows by noticing that
	$\lim_{\substack{\zeta \to 0\\ \zeta > 0}} \psi'(\zeta) = -\infty$
	and $\lim_{\zeta \to \pinf} \psi'(\zeta) = \pinf$. 
\end{proof}
\fi

As for the Renyi divergence, \eqref{e:opteqIalpha} becomes a
polynomial equation when $\alpha$ is a rational number.

\begin{remark}
	We can also derive the proximity operator of
	\begin{equation}
	\widetilde{\Phi}(\upsilon,\xi) =
	\left\{
	\begin{aligned}
	&\kappa\big(\alpha \upsilon+(1-\alpha) \xi\big) - \upsilon^\alpha\xi^{1-\alpha}
	&& \mbox{\small if $(\upsilon,\xi)\!\in\! \RP^2$}\\
	&\pinf && \mbox{otherwise,}
	\end{aligned}
	\right.
	\end{equation}
	where $\kappa \in \RR$. From Proposition~\ref{p:propprox}\ref{p:propproxiv}, we get, for every $\gamma \in \RPP$ and for every $(\overline{\upsilon},\overline{\xi})\in \RR^2$,
	\begin{equation}
	\prox_{\gamma \widetilde{\Phi}}(\overline{\upsilon},\overline{\xi})=
	\prox_{\gamma \Phi}\big(\overline{\upsilon}+\gamma(1-\kappa)\alpha,\overline{\xi}+\gamma(1-\kappa)(1-\alpha)\big),
	\end{equation}
	where $\prox_{\gamma\Phi}$ is provided by Proposition \ref{prop:proxIalpha}.
\end{remark}

\begin{table*}[t]
	\caption{Conjugate function $\varphi^*$ of the restriction of $\varphi$ to $\RP$.\label{t:conjFunc}}
	\begin{center}
		\begin{tabular}{l @{\qquad\qquad} l @{\qquad\qquad} l}
			\toprule
			\multirow{2}{*}{Divergence}  & \multicolumn{1}{l}{$\varphi(\zeta)$}  & \multicolumn{1}{l}{$\varphi^*(\zeta^*)$} \\
			& \multicolumn{1}{l}{$\zeta > 0$}  & \multicolumn{1}{l}{$\zeta^* \in \RR$} \\
			\midrule
			Kullback-Leibler & $\zeta \ln \zeta-\zeta+1$ &  $e^{\zeta^*}-1$\\
			\\
			Jeffreys & $(\zeta-1)\ln\zeta$& $\displaystyle
			W(e^{1-\zeta^*}) + \big(W(e^{1-\zeta^*})\big)^{-1}+\zeta^*-2$\\
			\\
			Hellinger & $1+\zeta-2\sqrt{\zeta}$ & $ \begin{cases}\displaystyle\frac{\zeta^*}{1-\zeta^*} & \mbox{if $\zeta^* < 1$} \\ \pinf & \mbox{otherwise} \end{cases}$\\
			\\
			Chi square & $ (\zeta-1)^2$ &  
			$\begin{cases}\displaystyle\frac{\zeta^*(\zeta^*+4)}{4} & \mbox{if $\zeta^* \ge -2$}\\-1 & \mbox{otherwise}\end{cases}$\\
			\\
			Renyi, $\alpha \in ]1,+\infty[$ & $\zeta^{\alpha}$ & 
			$\begin{cases}\displaystyle (\alpha-1)\Big(\frac{\zeta^*}{\alpha}\Big)^{\frac{\alpha}{\alpha-1}} & \mbox{if $\zeta^* \ge 0$}\\ 0 &\mbox{otherwise} \end{cases}$ \\
			\\
			I$_{\alpha}$, $\alpha \in ]0,1[$ & $1-\alpha+\alpha \zeta-\zeta^{\alpha}$ & 
			$\begin{cases}\displaystyle (1-\alpha) \Big(\Big(1-\frac{\zeta^*}{\alpha}\Big)^{\frac{\alpha}{\alpha-1}}
			-1\Big) & \mbox{if $\zeta^* \le \alpha$}\\ \pinf & \mbox{otherwise}\end{cases}$\\
			\bottomrule
			\\
		\end{tabular}
	\end{center}
\end{table*}

\section{Connection with epigraphical projections}%
\label{sec:epigraphicalProj}
{\color{black}
	Proximal methods iterate a sequence of steps in which proximity operators are evaluated. The efficient computation of these operators is thus essential for dealing with high-dimensional convex optimization problems. In the context of constrained optimization, at least one of the additive terms of the global cost to be minimized consists of the indicator function of a closed convex set, whose proximity operator reduces to the projections onto this set. However, if we except a few well-known cases, such projection does not admit a closed-form expression. The resolution of large-scale optimization problems involving non trivial constraints is thus quite challenging.
	This difficulty can be circumvented when the constraint can be expressed as the lower-level set of some separable function, by making use of epigraphical projection techniques.
	Such approaches have attracted interest in the last years \cite{Chierchia_2012_Epigraphical_Projection, steidl_Pesquet_epigraphicalProj_2013, Tofighi_2014_preprint_sig_recon_pesc, Ono_2014_icassp_TGV_constraint, Wang2016, Moerkotte2015}. The idea consists of decomposing the constraint of interest into the intersection of a half-space and a number of epigraphs of simple functions. For this approach to be successful, it is mandatory that the projection onto these epigraphs can be efficiently computed.
	
	The next proposition shows that the expressions of the projection onto the epigraph 
	of a wide range of functions can be deduced from the expressions of the proximity operators of
	$\varphi$-divergences. In particular, in Table~\ref{t:conjFunc}, for each of the $\varphi$-divergences presented in Section~\ref{s:mainresult},
	we list the associated functions $\varphi^*$ for which such projections can thus be derived.
}

\begin{proposition}
	Let $\varphi\colon\RR \to \RPX$ be a function in $\Gamma_0(\RR)$ which is twice differentiable on $\RPP$.
	Let $\Phi$ be the function defined by \eqref{e:perspf} and $\varphi^*\in \Gamma_0(\RR)$ the Fenchel-conjugate function of 
	the restriction of $\varphi$ on $\RP$, defined as
	\begin{equation}
	(\forall \zeta^* \in\RR) \qquad \varphi^*(\zeta^*) = \sup_{\zeta\in \RP} \zeta \zeta^* - \varphi(\zeta).
	\end{equation}
	Let the epigraph of $\varphi^*$ be defined as 
	\begin{equation}
	\epi \varphi^*=\menge{(\upsilon^*,\xi^*)\in \RR^2}{\varphi^*(\upsilon^*)\le \xi^*}.
	\end{equation}
	Then, the projection onto $\epi \varphi^*$ is: for every $(\upsilon^*,\xi^*)\in \RR^2$,  
	\begin{equation}\label{e:projepi}
	\mathsf{P}_{\epi \varphi^*}(\upsilon^*,\xi^*)=(\upsilon^*,-\xi^*)
	- \prox_{\Phi}(\upsilon^*,-\xi^*).
	\end{equation}
\end{proposition}

\begin{proof}
	The conjugate function of $\Phi$ is, for every $(\upsilon,\xi)\in \RR^2$,
	\begin{equation}
	\Phi^*(\upsilon^*,\xi^*) =
	\sup_{(\upsilon,\xi)\in \RR^2} \upsilon \upsilon^* + \xi \xi^* - \Phi(\upsilon,\xi).
	\end{equation}
	From the definition of $\Phi$, we deduce that, for all $(\upsilon,\xi)\in \RR^2$,
	\begin{align}
	\iflong\Phi^*(\upsilon^*,\xi^*)\fi
	&=\sup\Big\{
	\sup_{(\upsilon,\xi)\in \RP\times ]0,\pinf[} \Big(\upsilon \upsilon^* + \xi \xi^* -\xi \varphi\Big(\frac{\upsilon}{\xi}\Big)\Big), \nonumber\\
	&\qquad\qquad\quad\;\;\;\sup_{\upsilon\in ]0,\pinf[} \Big(\upsilon \upsilon^* -\lim_{\substack{\xi \to 0\\\xi > 0}} \xi \varphi\Big(\frac{\upsilon}{\xi}\Big)\Big),
	0\Big\}\\
	&=\sup\Big\{
	\sup_{(\upsilon,\xi)\in \RP\times ]0,\pinf[} \Big(\upsilon \upsilon^* + \xi \xi^* -\xi \varphi\Big(\frac{\upsilon}{\xi}\Big)\Big),
	0\Big\}\\
	&=\sup\{\iota_{\epi \varphi^*}(\upsilon^*,-\xi^*),0\}\label{e:resintconj}\\
	&= \iota_{\epi \varphi^*}(\upsilon^*,-\xi^*),
	\end{align}
	where the equality in \eqref{e:resintconj} stems from \cite[Example 13.8]{Bauschke_H_2011_book_con_amo}.
	Then, \eqref{e:projepi} follows from the conjugation property of the proximity operator (see Proposition \ref{p:propprox} \ref{p:propproxvi}).
\end{proof}


\section{Experimental results}\label{s:results}
To illustrate the potential of our results, we consider a query optimization problem in database management systems where the optimal query execution plan depends on the accurate estimation of the proportion of tuples that satisfy the predicates in the query. More specifically, every request formulated by a user can be viewed as an event in a probability space $(\Omega,\mathcal{T},\mathcal{P})$, where $\Omega$ is a finite set of size $N$. In order to optimize request fulfillment, it is useful to accurately estimate the probabilities, also called \emph{selectivities}, associated with each element of $\Omega$. To do so, rough estimations of the probabilities of a certain number $P$ of events can be inferred from the history of formulated requests and some a priori knowledge.

Let $x=(x^{(n)})_{1\le n\le N}\in\RR^N$ be the vector of sought probabilities, and let $z=(z^{(i)})_{1\le i\le P} \in [0,1]^P$ be the vector of roughly estimated probabilities. The problem of selectivity estimation is equivalent to the following constrained entropy maximization problem \cite{Markl2007}:
\begin{equation}\label{ex1:baseline}
\minimize{x\in\RR^N} \sum_{n=1}^N x^{(n)} \ln x^{(n)}
\quad\operatorname{s.t.}\quad
\left\{
\begin{aligned}
&Ax=z,\\
&\sum_{n=1}^N x^{(n)} = 1,\\
&x \in [0,1]^N,
\end{aligned}
\right.
\end{equation}
where $A\in\RR^{P\times N}$ is a binary matrix establishing the theoretical link between the probabilities of each event and the probabilities of the elements of $\Omega$ belonging to it.

Unfortunately, due to the inaccuracy of the estimated probabilities, the intersection between the affine constraints $Ax=z$ and the other ones may be empty, making the above problem infeasible. In order to overcome this issue, we propose to jointly estimate the selectivities and the feasible probabilities. Our idea consists of reformulating Problem \eqref{ex1:baseline} by introducing the divergence between $Ax$ and an additional vector $y$ of feasible probabilities. This allows us to replace the constraint $Ax=z$ with an $\ell_k$-ball centered in $z$, yielding
\begin{gather}
\minimize{(x,y)\in\RR^N\times\RR^P} D(Ax,y) + \lambda\sum_{n=1}^N x^{(n)} \ln x^{(n)} 
\quad\operatorname{s.t.}\quad
\left\{
\begin{aligned}
&\textcolor{black}{\|y-z\|_k\le \eta,}\\
&\sum_{n=1}^N x^{(n)} = 1,\\
&x \in [0,1]^N,
\end{aligned}
\right.\label{ex1:proposed}
\end{gather}
where $D$ is defined in \eqref{e:defD}, $\lambda$ and $\eta$ are some positive constants, whereas $k\in [1,\pinf[$ (the choice $k=2$ yields the Euclidean ball).

To demonstrate the validity of this approach, we compare it with the following methods:
\begin{enumerate}
\item a relaxed version of Problem \eqref{ex1:baseline}, in which the constraint $Ax=z$ is replaced with a squared Euclidean distance:
\begin{gather}
\minimize{x\in\RR^N} \|Ax-z\|^2+\lambda\sum_{n=1}^N x^{(n)} \ln x^{(n)} 
\quad\operatorname{s.t.}\quad
\left\{
\begin{aligned}
&\sum_{n=1}^N x^{(n)} = 1,\\
&x \in [0,1]^N,
\end{aligned}
\right.\label{ex1:relaxed}
\end{gather}
\textcolor{black}{
or with $\varphi$-divergence $D$:
\begin{gather}
\minimize{x\in\RR^N} D(Ax,z)+\lambda\sum_{n=1}^N x^{(n)} \ln x^{(n)} 
\quad\operatorname{s.t.}\quad
\left\{
\begin{aligned}
&\sum_{n=1}^N x^{(n)} = 1,\\
&x \in [0,1]^N,
\end{aligned}
\right.
\label{ex1:relaxed_div}
\end{gather}
}
where $\lambda$ is some positive constant;

\item the two-step procedure in \cite{Moerkotte2015}, which consists of finding a solution $\widehat{x}$ to 
\begin{equation}\label{ex1:preprocess}
\minimize{x\in\RR^N} Q_1\big(Ax, z \big)
\quad\operatorname{s.t.}\quad
\left\{
\begin{aligned}
&\sum_{n=1}^N x^{(n)} = 1,\\
&x \in [0,1]^N,
\end{aligned}
\right.
\end{equation}
and then solving \eqref{ex1:baseline} by replacing $z$ with $\widehat{z} = A\widehat{x}$. Hereabove, for every $y\in\RR^P$, $Q_1(y,z) = \sum_{i=1}^P \phi(y^{(i)}/z^{(i)})$ is a sum of quotient functions, i.e.
\begin{equation}
\phi(\xi) = 
\begin{cases}
\xi, &\textrm{if $\xi\ge 1$},\\
\xi^{-1}, &\textrm{if $0< \xi < 1$},\\
+\infty, &\textrm{otherwise}.\\
\end{cases}
\end{equation}
\end{enumerate}

For the numerical evaluation, we adopt an approach similar to \cite{Moerkotte2015}, and we first consider the following low-dimensional setting:
\begin{equation}
A = 
\begin{bmatrix}
1 & 0 & 1 & 0 & 1 & 0 & 1 \\
0 & 1 & 1 & 0 & 0 & 1 & 1 \\
0 & 0 & 0 & 1 & 1 & 1 & 1 \\
0 & 0 & 1 & 0 & 0 & 0 & 1 \\
0 & 0 & 1 & 0 & 1 & 0 & 1 \\
0 & 0 & 0 & 0 & 0 & 1 & 1
\end{bmatrix},
\qquad
z = 
\begin{bmatrix}
0.2114 \\
0.6331 \\
0.6312 \\
0.5182 \\
0.9337 \\
0.0035
\end{bmatrix},
\end{equation} 
for which there exists no $x\in \RP^N$ such that $Ax=z$. To assess the quality of the solutions $x^*$ obtained with the different methods, we evaluate the max-quotient between $Ax^*$ and $z$, that is \cite{Moerkotte2015}
\begin{equation}
Q_\infty(Ax^*, z) = \max_{1\le i\le P}\; \phi\left(\frac{[Ax^*]^{(i)}}{z^{(i)}}\right).
\end{equation}

Table \ref{tab:ex1} collects the $Q_\infty$-scores (lower is better) obtained with the different approaches. For all the considered $\varphi$-divergences,\footnote{Note that the Renyi divergence is not suitable for the considered application, because it tends to favor sparse solutions.} the proposed approach performs favorably with respect to the state-of-the-art, the KL divergence providing the best performance among the panel of considered $\varphi$-divergences. For the sake of fairness, the hyperparameters $\lambda$ and $\eta$ were hand-tuned in order to get the best possible score for each compared method. The good performance of our approach is related to the fact that $\varphi$-divergences are well suited for the estimation of probability distributions. 

Figure~\ref{fig:times} next shows the computational time for solving Problem~\eqref{ex1:proposed} for various dimensions $N$ of the selectivity vector to be estimated, with $A$ and $z$ randomly generated so as to keep the ratio $N/P$ equal to $7/6$. To make this comparison,  the primal-dual proximal method recalled in Algorithm~\ref{e:FBF} was implemented in MATLAB R2015, by using the stopping criterion $\|x_{n+1} - x_n\| < 10^{-7} \|x_n\|$. We then measured the execution times on an Intel i5 CPU at 3.20 GHz with 12 GB of RAM. The results show that all the considered $\varphi$-divergences can be efficiently optimized, with no significant computational time differences between them.

\begin{table}
	\caption{Comparison of $Q_\infty$-scores}
	\label{tab:ex1}
	\centering
	\textcolor{black}{
		\begin{tabular}{ccccc}
			\toprule
			\multicolumn{2}{c}{Problem \eqref{ex1:proposed}} & \eqref{ex1:relaxed_div} & \eqref{ex1:preprocess}+\eqref{ex1:baseline} \cite{Moerkotte2015} & \eqref{ex1:relaxed}\\
			$\varphi \Big.$ &  & & \\
			\midrule
			KL  & \textbf{2.23} & 2.95 & \multirow{4}{*}{2.45} & \multirow{4}{*}{25.84} \\
			Jef  & 2.44          & 3.41 \\
			Hel & 2.42          & 89.02 \\
			Chi & 2.34          & 3.20  \\
			I$_{{1}/{2}}$ & 2.42& 89.02  \\
			\bottomrule
		\end{tabular}
	}
\end{table}

\begin{figure}
	\centering
	\includegraphics[width=0.45\textwidth]{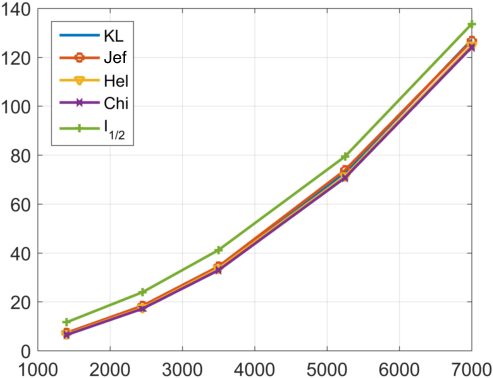}
	\caption{Execution time (in seconds) versus size $N$ in Problem \eqref{ex1:proposed}.}
	\label{fig:times}
\end{figure}

\section{Conclusion}
\label{sec:conclch3}
In this paper, we have shown how to solve convex optimization problems involving discrete information divergences by using proximal methods. We have carried out a thorough study of the properties of the proximity operators of $\varphi$-divergences, which has led us to derive new tractable expressions of them. In addition, we have related these expressions to the projection onto the epigraph of a number of convex functions.

Finally, we have illustrated our results on a selectivity estimation problem, where $\varphi$-divergences appear to be well suited for the estimation of the sought probability distributions. Moreover, computational time evaluations allowed us to show that the proposed numerical methods provide efficient solutions for solving large-scale optimization problems.

\iflong
\appendix
\fi

\iflong
\section{Proof of Proposition \ref{p:proxdifsimp}}\label{a:proxdifsimp}
Let $(\overline{\nu},\overline{\xi})\in\RR^2$.
From Proposition~\ref{p:propprox}\ref{p:propproxi}, we know that $\prox_{\Phi}(\overline{\nu},\overline{\xi})\in \RP^2$.
By using Proposition~\ref{p:propprox}\ref{p:propproxii}, we have the following equivalences:
\begin{align}\label{e:difcritder}
\begin{cases}
(\nu,\xi) \in \RPP^2\\
(\nu,\xi) = \prox_{\Phi}(\overline{\nu},\overline{\xi})
\end{cases}
&\Leftrightarrow
\begin{cases}
(\nu,\xi) \in \RPP^2\\
\overline{\nu}-\nu \in \partial \varphi(\nu-\xi)\\
\overline{\xi}-\xi \in -\partial\varphi(\nu-\xi)
\end{cases}\nonumber\\
&\Leftrightarrow
\begin{cases}
(\nu,\xi) \in \RPP^2\\
(\nu,\xi) = \prox_{\tilde{\Phi}}(\overline{\nu},\overline{\xi}),
\end{cases}
\end{align}
where $\tilde{\Phi}\colon (\nu,\xi) \mapsto \varphi(\nu-\xi)$.
By using now Proposition~\ref{p:propprox}\ref{p:propproxv}, we get
\begin{align}
	\eqref{e:difcritder}
	&\Leftrightarrow
	\begin{cases}
	(\nu,\xi) \in \RPP^2\\
	(\nu,\xi) = (\overline{\nu},\overline{\xi})+ \frac12 \big(\prox_{2\varphi}(\overline{\nu}-\overline{\xi})- \overline{\nu}+\overline{\xi}\big) (1,-1)
	\end{cases}\nonumber\\
	&\Leftrightarrow
	\begin{cases}
	(\nu,\xi) \in \RPP^2\\
	\nu = \frac12\big(\overline{\nu}+\overline{\xi}+\prox_{2\varphi}(\overline{\nu}-\overline{\xi})\big)\\
	\xi =\frac12\big(\overline{\nu}+\overline{\xi}-\prox_{2\varphi}(\overline{\nu}-\overline{\xi})\big)
	\end{cases}\nonumber\\
	&\Leftrightarrow
	\begin{cases}
	\left|\prox_{2\varphi}(\overline{\nu}-\overline{\xi})\right| < \overline{\nu}+\overline{\xi}\\
	\nu = \frac12\big(\overline{\nu}+\overline{\xi}+\prox_{2\varphi}(\overline{\nu}-\overline{\xi})\big)\\
	\xi =\frac12\big(\overline{\nu}+\overline{\xi}-\prox_{2\varphi}(\overline{\nu}-\overline{\xi})\big).
	\end{cases}
\end{align}
Similarly, we have
\begin{align}
	&\quad\;\begin{cases}
	\nu = 0, \xi \in \RPP\\
	(\nu,\xi) = \prox_{\Phi}(\overline{\nu},\overline{\xi})
	\end{cases}\nonumber\\
	&\Leftrightarrow
	\begin{cases}
	\nu = 0, \xi \in \RPP\\
	\overline{\nu} - \varphi'(-\xi) \in \partial \iota_{\RP}(0)=]-\infty,0]\\
	\xi-\overline{\xi} - \varphi'(-\xi) = 0
	\end{cases}\nonumber\\
	&\Leftrightarrow
	\begin{cases}
	\nu = 0, \xi \in \RPP\\
	\xi \ge \overline{\nu}+\overline{\xi}\\
	\xi = \prox_{\varphi(-\cdot)}\overline{\xi}
	\end{cases}\nonumber\\
	&\Leftrightarrow
	\begin{cases}
	\prox_{\varphi(-\cdot)}\overline{\xi}\in \RPP\cap [\overline{\nu}+\overline{\xi},\pinf[\\
	\nu = 0\\
	\xi = \prox_{\varphi(-\cdot)}\overline{\xi}
	\end{cases}\nonumber\\
	&\Leftrightarrow
	\begin{cases}
	\prox_{\varphi}\overline{\xi}\in \RPP\cap [\overline{\nu}+\overline{\xi},\pinf[\\
	\nu = 0\\
	\xi = \prox_{\varphi}\overline{\xi},
	\end{cases}
\end{align}
where the last equivalence results from the assumption that $\varphi$ is even.
Symmetrically,
\begin{equation}
	\begin{cases}
	\nu \in \RPP, \xi=0\\
	(\nu,\xi) = \prox_{\Phi}(\overline{\nu},\overline{\xi})
	\end{cases}
	\!\!\!\!\!\Leftrightarrow
	\begin{cases}
	\prox_{\varphi}\overline{\nu}\in \RPP\cap [\overline{\nu}+\overline{\xi},\pinf[\\
	\nu = \prox_{\varphi}\overline{\nu}\\
	\xi = 0.
	\end{cases}
\end{equation}
\fi

\section{Convergence proof of Algorithm \ref{algo:Newton_DKL}}\label{s:algoNewton_DKL}
We aim at finding the unique zero on $]\exp(-\gamma^{-1}\overline{\upsilon}),+\infty[$ of the function $\psi'$ given 
by \eqref{e:derpsiKL} along with its derivatives:
\begin{align}
(\forall \zeta\in \RPP)\quad \psi''({\zeta}) &= 1+\ln\zeta + \gamma^{-1}\overline{\upsilon} + \zeta^{-2},\\[0.5em]
\psi'''({\zeta}) &= \zeta^{-1} - 2\,\zeta^{-3}.
\end{align}
To do so, we employ the Newton method given in Algorithm~\ref{algo:Newton_DKL}, the convergence of which is here established. Assume that
\begin{itemize}
	\item $\Big.\left(\overline{\upsilon},\overline{\xi}\right)\in\RR^2$ are such that $\exp(\gamma^{-1}\overline{\upsilon}) > 1-\gamma^{-1}\overline{\xi}$,
	\item $\Big.\widehat{\zeta}$ is the zero on $\left]\exp(-\gamma^{-1}\overline{\upsilon}),+\infty\right[$ of $\psi'$,
	\item $\Big.(\widehat{\zeta}_n)_{n\in\NN}$ is the sequence generated by Algorithm~\ref{algo:Newton_DKL},
	\item $\epsilon_n = \widehat{\zeta}_n - \widehat{\zeta}\Big.$ for every $n\in\NN$.
\end{itemize}
We first recall a fundamental property of the Newton method, and then we proceed to the actual convergence proof.

\begin{lemma}\label{th:lemma_newton}
	For every $n\in\NN$, 
	\begin{equation}
	\epsilon_{n+1} = \epsilon_{n}^2 \frac{\psi'''(\varrho_n)}{2\psi''(\widehat{\zeta}_n)}
	\end{equation}
	where $\varrho_n$ is between $\widehat{\zeta}_n$ and $\widehat{\zeta}$.
\end{lemma}
\begin{proof}
	The definition of $\epsilon_{n+1}$ yields
	\begin{equation}
	\epsilon_{n+1} = \widehat{\zeta}_n - \frac{\psi'({\widehat{\zeta}_n})}{\psi''({\widehat{\zeta}_n})} -\widehat{\zeta} = \frac{\epsilon_n\psi''({\widehat{\zeta}_n})-\psi'({\widehat{\zeta}_n})}{\psi''({\widehat{\zeta}_n})}.
	\end{equation}
	Moreover, for every $\widehat{\zeta}_n \in \RPP$, the second-order Taylor expansion of $\psi'$ around $\widehat{\zeta}_n$ is
	\begin{equation}
	\psi'(\widehat{\zeta}) = \psi'(\widehat{\zeta}_n) + \psi''(\widehat{\zeta}_n) (\widehat{\zeta}-\widehat{\zeta}_n) + \frac{1}{2} \psi'''(\varrho_n)(\widehat{\zeta}-\widehat{\zeta}_n)^2,
	\end{equation}
	where $\varrho_n$ is between $\widehat{\zeta}_n$ and $\widehat{\zeta}$. From the above equality, we deduce that $\psi'(\widehat{\zeta}) = \psi'(\widehat{\zeta}_n) - \psi''(\widehat{\zeta}_n) \epsilon_n + \frac{1}{2} \psi'''(\varrho_n)\epsilon_n^2=0$.
\end{proof}

\begin{proposition}
The sequence $(\widehat{\zeta}_n)_{n\in\mathbb{N}}$ converges to $\widehat{\zeta}$.
\end{proposition}
\begin{proof}
	The assumption $\exp(\gamma^{-1}\overline{\upsilon}) >
	1-\gamma^{-1}\overline{\xi}$ implies that $\psi'$ is negative at the initial value $\widehat{\zeta}_0=\exp(-\gamma^{-1}\overline{\upsilon})$, that is 
	\begin{equation}
	\psi'(\widehat{\zeta}_0) = - \exp(\gamma^{-1}\overline{\upsilon}) + 1 - \gamma^{-1}\overline{\xi}
	<0.
	\end{equation} 
	Moreover, $\psi'$ is increasing on $\left[\exp(-\gamma^{-1}\overline{\upsilon}),+\infty\right[$, since
	\begin{equation}
	\big(\forall {\zeta}\in\left[\exp(-\gamma^{-1}\overline{\upsilon}),+\infty\right[\big)\qquad
	\psi''({\zeta}) >0,
	\end{equation}
	and $\sqrt{2}$ is a non-critical inflection point for $\psi'$, since
	\begin{align}
	\big(\forall {\zeta}\in\big]\sqrt{2},+\infty\big[\big)\qquad& 
	\psi'''({\zeta}) > 0,\\[0.5em]
	\big(\forall {\zeta}\in\big]0,\sqrt{2}\big[\big)\qquad& 
	\psi'''({\zeta}) < 0.
	\end{align}
	To prove the convergence, we consider the following cases:
	\begin{itemize}
		\item \emph{Case $\widehat{\zeta} \le \sqrt{2}$}: $\psi'$ is increasing and concave on $[\widehat{\zeta}_0,\sqrt{2}]$. Hence, Newton method initialized at the lower bound of interval $[\widehat{\zeta}_0,\widehat{\zeta}]$ monotonically increases to $\widehat{\zeta}$ \cite{kincaid2002_book}.

		\item \emph{Case $\sqrt{2} \le \widehat{\zeta}_0 < \widehat{\zeta}$}: $\psi'$ is increasing and convex on $[\widehat{\zeta}_0,+\infty[$. Hence, Lemma~\ref{th:lemma_newton} yields $\epsilon_{1} = \widehat{\zeta}_{1} - \widehat{\zeta} > 0$. It then follows from standard properties of Newton algorithm for minimizing an increasing convex function that	$(\widehat{\zeta}_n)_{n\ge1}$ monotonically decreases to $\widehat{\zeta}$ \cite{kincaid2002_book}.
		
		\item \emph{Case $\Big.\widehat{\zeta}_0 < \sqrt{2} < \widehat{\zeta}$}: as $\psi'$ is negative and increasing on $[\widehat{\zeta}_0,\widehat{\zeta}[$, the quantity $-\psi' / \psi''$ is positive and lower bounded on $[\widehat{\zeta}_0,\sqrt{2}]$, that is
		\begin{equation}
		\big(\forall{\zeta}\in[\widehat{\zeta}_0,\sqrt{2}]\big)\quad -\frac{\psi'({{\zeta}})}{\psi''({{\zeta}})} \ge - \frac{\psi'({\sqrt{2}})}{\psi''({\widehat{\zeta}_0})} > 0.
		\end{equation}
		There thus exists $k > 0$ such that $\widehat{\zeta}_0 <  \ldots < \widehat{\zeta}_k$ and  $\widehat{\zeta}_{k} > \sqrt{2}$. Then, the convergence of $(\widehat{\zeta}_n)_{n\ge k}$ follows from the same arguments
		as in the previous case.
	\end{itemize}
\end{proof}

\iflong
\section{Convergence proof of Algorithm \ref{algo:Newton}}
\label{s:algoNewton}
We aim at finding the unique zero on $]W(e^{1-\gamma^{-1}\overline{\upsilon}}),\pinf[$ of the function $\psi'$ given 
by \eqref{e:derJK}, whose derivative reads
\begin{equation}
(\forall \zeta\in \RPP)\quad \psi''(\zeta)  = \ln\zeta + \frac{1}{\zeta} + \frac{1}{\zeta^2} + 2\zeta + \gamma^{-1}\overline{\upsilon}.
\end{equation}
To do so, we employ the projected Newton algorithm, whose global convergence is guaranteed for any initial value by the following condition \cite{global_conv_petersen_2004_JMMOR}: $(\forall a \in ]0,+\infty[)(\forall b \in ]a,+\infty[)$
\begin{equation}
\qquad \psi''(a)+\psi''(b) > \frac{\psi'(b)-\psi'(a)}{b-a},
\end{equation}
which is equivalent to
\begin{equation}\label{n:eq_conv_cond}
(b-a) \psi''(a) + (b-a) \psi''(b) - \psi'(b) + \psi'(a) > 0.
\end{equation}
Condition \eqref{n:eq_conv_cond} can be rewritten as follows
\begin{align}
&(b-a) (\ln a + \frac{1}{a} + \frac{1}{a^2} + 2a + \gamma^{-1}\overline{\upsilon}) \nonumber \\
&+ (b-a) (\ln b + \frac{1}{b} + \frac{1}{b^2} + 2b + \gamma^{-1}\overline{\upsilon}) \nonumber  \\
&+\Big(  (a+1)\ln a + \frac{a}{2} + 1 - \frac{1}{a} + a^2 + \Big(\gamma^{-1}\overline{\upsilon}- \frac{3}{2}\Big)a - \gamma^{-1} \overline{\xi} \Big)\nonumber  \\
&- \Big( (b+1)\ln b + \frac{b}{2} + 1 - \frac{1}{b} + b^2 + \Big(\gamma^{-1}\overline{\upsilon}- \frac{3}{2}\Big)b - \gamma^{-1} \overline{\xi} \Big)\nonumber\\
&> 0,
\end{align}
which, after some simplification, boils down to
\begin{align}\label{n:eq_newExp_psi2}
&(b+1)\ln a - (a+1)\ln b + \frac{b}{a} + \frac{b}{a^2} - \frac{2}{a} - \frac{a}{b} - \frac{a}{b^2}  \nonumber\\
&+ \frac{2}{b} - a^2 + b^2+ \gamma^{-1}\overline{\upsilon}(b-a) + \frac{1}{2}(b-a) > 0.
\end{align}
We now show that Condition \eqref{n:eq_newExp_psi2} holds true because $b>a$ and it is a sum of two terms  
\begin{align}
\label{n:eq_gab_gba}
& (b+1)\ln a - (a+1)\ln b + \frac{b}{a} + \frac{b}{a^2} - \frac{2}{a} - \frac{a}{b} - \frac{a}{b^2}\nonumber \\
& + \frac{2}{b} - a^2 + b^2
>0
\end{align}
and 
\begin{align}
\underbrace{\gamma^{-1}\overline{\upsilon}(b-a) + \frac{1}{2}(b-a)}_{>0}.
\end{align}

Indeed, \eqref{n:eq_gab_gba} can be rewritten as
\begin{equation*}
(\forall a \in ]0,+\infty[)(\forall b \in ]a,+\infty[)\qquad g(a,b)-g(b,a) > 0
\end{equation*}
where
\begin{equation*}
g(x,y) = -(x+1)\ln y  - \frac{x}{y}  + \frac{y}{x^2}  - \frac{2}{x} + y^2.
\end{equation*}
Therefore, we shall demonstrate that, for every $b>a>0$, $g$ is decreasing w.r.t.\ the first argument and increasing w.r.t.\ to the second argument,  i.e.
\begin{equation}\label{n:ineq_arg}
g(a,b)>g(b,b) \qquad\textrm{and}\qquad g(b,b)>g(b,a),
\end{equation} 
which implies that
\begin{equation*}
g(a,b)>g(b,a).
\end{equation*} 
To prove these two inequalities, we will study the derivative of $g$ with respect to its arguments. The conditions in \eqref{n:ineq_arg} are indeed equivalent to 
\begin{align}
& (\forall y \in ]0,+\infty[ ) (\forall x \in ]0,y[)  && \frac{\partial g }{\partial x}(x,y)<0,\label{n:ineq_1arg}\qquad \qquad \qquad \;\;\;\\
& (\forall x \in ]0,+\infty[ ) (\forall y \in ]x,+\infty[)  && \frac{\partial g }{\partial y}(x,y)>0.\label{n:ineq_2arg}\qquad \qquad \qquad \;\;\;
\end{align}
The first and second partial derivatives of $g$ w.r.t.\ $x$ read\\
$(\forall y \in ]0,+\infty[ ) (\forall x \in ]0,y[)$
\begin{align*}
& \frac{\partial g}{\partial x}(x,y)  = -\ln y -\frac{1}{y}-\frac{2y}{x^3}+\frac{2}{x^2}\\
&\frac{\partial^2 g}{\partial x^2}(x,y)  = \frac{6y}{x^4}-\frac{4}{x^3} = \frac{6y-4x}{a^4}>0. 
\end{align*}
Since $\partial^2 g/\partial x^2$ is strictly positive, $\partial g/\partial x$ is strictly increasing w.r.t.\ $x$ and
\begin{align*}
& \lim_{x \rightarrow y}  \frac{\partial g}{\partial x}(x,y)  = -\ln y -\frac{1}{y} = \ln \frac{1}{y} -\frac{1}{y} < 0.
\end{align*}
Therefore, Condition \eqref{n:ineq_1arg} holds, and $g$ is decreasing with respect to $x$. \\
The first and second partial derivatives of $g$ w.r.t.\ $y$ read\\
$(\forall x \in ]0,+\infty[ ) (\forall y \in ]x,+\infty[)$
\begin{align*}
\qquad & \frac{\partial g}{\partial y}(x,y) = \frac{-x}{y}-\frac{1}{y}+\frac{x}{y^2}+\frac{1}{x^2}+2y\\
\qquad& \frac{\partial^2 g}{\partial y^2}(x,y) = \frac{x}{y^2}+\frac{1}{y^2} - \frac{2x}{y^3} +2.
\end{align*}
For every $y \in [1,+\infty[$,
\begin{equation*}
(\forall x \in ]0,y[ ) \qquad \frac{\partial^2 g}{\partial y^2}(x,y) = \frac{x}{y^2}+\frac{1}{y^2} - \underbrace{\frac{2x}{y^3}}_{<1} +2 >0,
\end{equation*}
and $\partial g/\partial y$ is strictly increasing w.r.t.\ $y$ (since $\partial^2 g/\partial y^2$ is strictly positive) and
\begin{align*}
(\forall x \in [1,+\infty[ )&& \quad \lim_{y \rightarrow x} \;\;  & \frac{\partial g}{\partial y}(x,y) = -1 + \frac{1}{x^2}+2x>0,\\
(\forall x \in ]0,1] )&&   &\frac{\partial g}{\partial y}(x,1) =  \frac{1}{x^2}+1>0.
\end{align*}
For every $y \in\; ]0,1[ $, we have
\begin{align*}
(\forall x \in ]0,y[ )  \qquad  \frac{\partial g}{\partial y}(x,y)& =\frac{-x}{y}+\frac{x}{y^2} -\frac{1}{y}+\frac{1}{x^2}+2y,\\
& =\frac{x-xy}{y^2} -\frac{1}{y}+\frac{1}{x^2}+2y,
\end{align*}
since $x<y<1$, this implies that $xy<x$ and $\frac{1}{y}<\frac{1}{x}<\frac{1}{x^2}$ and 
\begin{align*}
(\forall x \in ]0,y[ )  \qquad  \frac{\partial g}{\partial y}(x,y)= \underbrace{\frac{x-xy}{y^2}}_{>0} \underbrace{-\frac{1}{y}+\frac{1}{x^2}}_{>0}+2y>0.
\end{align*}
As Condition \eqref{n:ineq_2arg} holds, $g$ is increasing with respect to $y$. 
\fi

\bibliographystyle{IEEEbib}
\bibliography{abbr,biblio}

\end{document}